\documentclass[3p]{elsarticle}
\usepackage[utf8]{inputenc}
\usepackage{graphicx}
\usepackage{amssymb}
\usepackage{amsmath}
\usepackage{amsthm}
\usepackage{amsfonts}
\usepackage{array}
\usepackage{float}
\usepackage{caption}
\usepackage[nottoc]{tocbibind}
\usepackage{esvect}
\usepackage{enumitem}
\usepackage{hyperref} 
\usepackage{epigraph} 
\usepackage{url}
\usepackage[all]{xy}
\usepackage{tikz-cd}
\usetikzlibrary{decorations.pathreplacing}
\usepackage{faktor}
\usepackage[dutch,english]{babel}
\usepackage{hhline}
\usepackage{relsize}
\usepackage{subcaption} 
\usepackage{verbatim} 

\newcommand{\dr}{\gamma} 
\newcommand{\tS}[1]{\tilde S_{#1}(-\dr)}
\newcommand{\tT}[1]{\tilde T_{#1}(-\dr)}
\newcommand{\crop}[1]{\mathcal{W}_{#1}}   

\newtheorem{lem}{Lemma}
\newtheorem{theorem}{Theorem}
\newtheorem{definition}{Definition}

\newcommand*\colvec[3][]{
    \begin{pmatrix}\ifx\relax#1\relax\else#1\\\fi#2\\#3\end{pmatrix}
}

\usepackage{scalerel}
\usepackage{stackengine,wasysym}

\newcommand\wtilde[1]{\ThisStyle{%
  \setbox0=\hbox{$\SavedStyle#1$}%
  \stackengine{-.1\LMpt}{$\SavedStyle#1$}{%
    \stretchto{\scaleto{\SavedStyle\mkern.2mu\AC}{.5150\wd0}}{.6\ht0}%
  }{O}{c}{F}{T}{S}%
}}

\title{$\gamma$-CounterBoost: Optimizing Response Time Tails using Job Type Information Only}

\author[1]{Nils Charlet}
\ead{nils.charlet@uantwerpen.be}

\author[1]{Benny Van Houdt}
\ead{benny.vanhoudt@uantwerpen.be}

\affiliation[1]{organization={Dept. Computer Science, University of Antwerp},
city={Antwerp},
 citysep={,}, 
country={Belgium}}

\begin{document}

\begin{abstract}
    In a recent paper the $\gamma$-Boost scheduling policy was shown to minimize the tail of the response time distribution in a light-tailed M/G/1-queue. This policy schedules jobs using a boosted arrival time, defined as the arrival time of a job minus its boost, where the boost of a job depends on its exact job size. The $\gamma$-Boost policy can also be used when only partial job size information is available, such as the type of an incoming job. In such case the boost $b_i$ of a job depends solely on its type $i$ and $\gamma$-Boost was shown to optimize the tail among all boost policies, where a boost policy is fully determined by the $b_i$ values. In the partial information setting $\gamma$-Boost relies on two types of information: job types and arrival times. 

    This paper focuses on the problem of minimizing the tail in a light-tailed M/G/1-queue in the partial job size information setting when the scheduler only makes use of the job types and {\it does not exploit arrival times}. Prior work showed that
    in case of $2$ job types the so-called Nudge-$M$ policy minimizes the tail in a large class of scheduling policies. In this paper we introduce the $\gamma$-CounterBoost policy in the partial information setting with $d \geq 2$ job types and prove that it minimizes the tail in an even broader class of scheduling policies called Contextual CounterBoost policies. The $\gamma$-CounterBoost policy reduces to the Nudge-$M$ policy in case of $d=2$ job types. 
\end{abstract}
\maketitle

\section{Introduction}

It is not hard to show that the mean response time of a job in a single server queue is minimized by the
Shortest Remaining Processing Time (SRPT) policy. To implement this policy exact information on the
processing time of a job is required. In the absence of such information, the Gittins index policy is known to minimize
the mean response time in the M/G/1 queue \cite{Gittins2011,Aalto2009,Scully2021}. The Gittins index policy does not require job size information of 
individual jobs, it only requires information about the job size distribution. When minimizing the
mean response time, the Gittins policy assigns an index
to each job based on the amount of service a job has received thus far and serves the job with the highest index.

The problem of optimizing the tail, instead of the mean, of the response
time distribution in an M/G/1 queue is a more difficult problem. For heavy tailed
distributions classic policies like
Processor Sharing, SRPT, Least Attained Service as well as the Gittins policy are tail optimal for 
{\it regularly-varying} distributions \cite{borst2003,boxma07,boostZiv}. For most light tailed distributions (so-called class I distributions
\cite{abate1}) 
it was long believed that the First Come First Served (FCFS) policy is tail optimal \cite{wierman12}. However it was recently shown that
the response time tail of FCFS can be improved upon using the Nudge scheduling policy \cite{nudge}. More specifically,  
the response time under FCFS in an M/G/1 queue with a class I job size distribution decays exponentially, meaning
\[ P[R_{FCFS} > t ] \approx c_{FCFS} e^{-\gamma t},\]
for some constants $c_{FCFS}$ and $\gamma$. The decay rate $\gamma$ of FCFS is known to be optimal (as it is the
decay rate of the workload in the queue which does not depend on the scheduling policy) \cite{boxma07}, but the {\it prefactor}
$c_{FCFS}$ can be further reduced by the Nudge scheduling policy \cite{nudge}. The initial Nudge scheduling policy operates
almost identical to FCFS, except that an arriving {\it short} job may pass/skip the last job in the queue,
provided that the last job is a {\it long} job that was not passed by any other job.

Intuitively, the difference between dealing with light and heavy tailed distributions lies in the
manner in which they cause long delays. In \cite{boxma07}, it is described that for heavy tailed job sizes, a large workload is most likely caused by a single huge job, one \textit{catastrophe} as they call it. For this reason policies that are tail optimal in this setting ensure that small jobs do not wait behind these huge jobs, even when this huge job is already in service. In contrast, for light tailed job sizes, long delays are mostly caused by a \textit{conspiracy}, which corresponds to a large period of time in which inter-arrival times are smaller than expected, while job sizes are larger than expected.

\subsection*{Tail Optimality for Light-Tailed (class I) jobs}

The discovery of the improved response time tail of the Nudge scheduling policy opened up many new research
directions as it became immediately clear that it is not hard to reduce the prefactor of the
initial Nudge scheduling policy by considering different Nudge variations \cite{vanhoudt_nudge, charlet2024nudgeM}. A key characteristic of these Nudge
variations is that they all rely solely on job size (or type) information of the last few jobs
that arrived in the queue to determine how many jobs an arriving job is allowed to pass. 
This job size/type information is too limited to minimize the prefactor of the tail of the response time among all scheduling policies.

It turns out that apart from the job sizes, the arrival times of jobs also play a key role in coming up with
a response time tail optimal policy.  This led to the introduction of the $\gamma$-Boost scheduling policy which was shown to have an optimal response time tail, that is, it minimizes the prefactor, for class I job size distributions in case the scheduler knows the exact size $z$ and arrival time of
any incoming job \cite{boostZiv}. The $\gamma$-Boost policy schedules jobs based on their {\it boosted arrival time}, as opposed to
their actual arrival time used by FCFS. 
The boosted arrival time of a job is defined as its actual arrival time $t_a$ minus its boost $b(z)$, which depends
on the exact job size $z$ of the job. The boost function $b(z)$ is defined as 
\[b(z)  = \frac{1}{\gamma}\log(1-\exp(-\gamma z)),\] 
where $\gamma$ is the decay rate of the response time in a standard M/G/1 queue (which uses the FCFS policy).  
The proof method exists in linking the problem of
minimizing the prefactor to minimizing the weighted discounted cost of a finite batch scheduling problem.

The $\gamma$-Boost policy can also be used when the job size information that is available to the scheduler is more coarse \cite{boostZiv}.
For instance, suppose jobs are partitioned into a number of types and the scheduler only knows the job type of an
incoming job. In such case the boosted arrival time is still defined as the actual arrival time minus the boost, but the
boost of a job now depends on its type $i$ as opposed to its exact job size $z$. The boost of a type-$i$ job 
when using $\gamma$-Boost is defined as
\[b_i  = \frac{1}{\gamma}\log \left( \frac{\tilde S_i(-\gamma) }{\tilde S_i(-\gamma)-1} \right),\] 
where $\tilde S_i(s)$ is the Laplace transform of the job size distribution of a type-$i$ job. In this coarser information setting, $\gamma$-Boost
was shown to achieve the lowest prefactor among all boost policies, meaning the choice of the boost value $b_i$
is optimal. In case of $2$ job types, $\gamma$-Boost was also shown to achieve a lower prefactor than the Nudge-$M$
policy (see further for a definition). The $\gamma$-Boost policy is believed to be tail optimal among all scheduling policies in this partial information setting, but a proof is still lacking. Furthermore, $\gamma$-Boost not only minimizes the tail
of the response time, but often offers a stochastic improvement over FCFS, meaning it is also relevant in practice.

\subsection*{Nudge Policies Boost Arrival Counters}
After the introduction of the initial Nudge policy, a number of Nudge variations were introduced that further reduced the prefactor in case of class I job size distributions. This was done in a context with $2$ job types, where the scheduler knows the type of an incoming job, but not its exact size or arrival time. This is in contrast to the initial Nudge policy which does make use of job sizes. The current state-of-the-art Nudge policy, initially called Nudge-$M$
\cite{charlet2024nudgeM}, simply states that
whenever a type-$1$ arrival joins the queue it is served before any type-$2$ job that (a) is still waiting in the queue and that
(b) was one of the previous $M$ arrivals (of any type). The Nudge-$M$ policy was proven to minimize the
prefactor in a large class of
Nudge policies if the parameter $M$ is set as 
\[ M_{opt} = 
\left\lfloor \frac{1}{\log \tilde S(-\gamma)} \log \left( \frac{\tilde S_1(-\gamma) (\tilde S_2(-\gamma)-1)}{\tilde S_2(-\gamma) (\tilde S_1(-\gamma)-1)}
\right) \right\rfloor,\]
provided that the types are labeled such that $\tilde S_1(-\gamma) \leq \tilde S_2(-\gamma)$, where $\tS{} = p \tS{1} + (1-p) \tS{2}$ is the Laplace transform of the job size distribution $S$ of a random job, and $p$ is the probability of a type-1 job.
This formula resembles the earlier expression for the boost $b_i$. In fact, $M_{opt}$ can be rewritten as 
\[ M_{opt} = 
\left\lfloor \frac{1}{\log \tilde S(-\gamma)}\log \left( \frac{\tilde S_1(-\gamma) }{\tilde S_1(-\gamma)-1} 
\right)-\frac{1}{\log \tilde S(-\gamma)} \log \left( \frac{\tilde S_2(-\gamma) }{\tilde S_2(-\gamma)-1} 
\right) \right\rfloor,\]
meaning if we define $b_i^* = \log (\tilde S_i(-\gamma)/(\tilde S_i(-\gamma)-1)) / \log \tilde S(-\gamma)$, then
$M_{opt} = \lfloor b_1^* - b_2^* \rfloor$. We state that the {\it arrival counter} of a job equals $n$ if it is the
$n$-th job arrival (throughout the paper we often use relative arrival counters as only the difference matters).
As any arriving type-$1$ job can pass up to $M$ type-$2$ jobs that were part of the last $M$ arrivals, we can therefore think
of the Nudge-$M$ policy as a scheduling policy that schedules jobs based on a {\it boosted arrival counter}, 
where $b^*_i$ is the counter boost of a type-$i$ job. Indeed, for two job types
this is the same as boosting type-$1$ jobs by $b_1^*-b_2^*$ and type-$2$ jobs by zero.
As such we refer to the Nudge-$M$ policy as a {\it CounterBoost} policy from here onward.

\subsection*{Challenges when working with $d > 2$ types}
The observation that the Nudge-$M$ policy in the context of $2$ job types can be regarded as a scheduling policy that uses
a boosted arrival counter immediately yields a natural generalization to a context with $d > 2$ job types 
(possibly even infinitely many). In case 
we have $d$ types, we simply define the counter boost of a type $i$ job as 
\[b_i^* = \frac{1}{\log \tS{}} \log \left(\frac{\tS{i}}{\tS{i}-1}\right),\] 
for all $i$. We refer to this CounterBoost policy as $\gamma$-CounterBoost.
This policy coincides with Nudge-$M$ in case of $2$ job types (when $M=M_{opt}$). 

The Nudge-$M$ policy was shown to minimize the
tail among a large class $\mathcal{F}$ of Nudge policies \cite{charlet2024nudgeM}. An obvious question is 
whether this result generalizes to $d > 2$ types. 
The class of Nudge policies $\mathcal{F}$ considered in the two job type setting was such that each policy 
in $\mathcal{F}$ was characterized by
a constant $M$ and some function $n$. The function $n$ maps each sequence $s \in \{1,2\}^M$ of job types to $n(s) \in \{0,\ldots,t(s)\}$, where $n(s)$ specifies how many type-$2$ jobs a type-$1$ job can pass when the types of the last $M$ arrivals are given by $s$ and $t(s)$ counts the number of twos in $s$. 
In this case one additional condition must hold on $n$ such that the policy characterized by $n$ is feasible. 
The feasibility is related to the fact that whenever some type-$1$ job with arrival time $t_a$ passes a type-$2$ job
with arrival time $t_b$, any type-$1$ job that arrives in $(t_b,t_a)$ must also pass the type-$2$ job.

A natural way to generalize this class of Nudge policies to $d > 2$ types would be to replace the function $n$ by a 
set of functions $n_{ij}$, where $n_{ij}(s)$ indicates how many type-$j$ jobs an arriving type-$i$ job can pass given
that the previous $M$ arrival types are given by $s$. A first difficulty that arises here is that the conditions needed
for feasibility are much more involved, even for $d=3$ types. To prove that $\gamma$-CounterBoost is optimal in this class of Nudge policies, one could nevertheless try to use the same proof technique as in the two type setting.
While some of the initial results in \cite{charlet2024nudgeM} can be generalized, the final step in the proof technique appears 
problematic. In this step one needs to find for each Nudge policy $\texttt{A}$ in this class a finite sequence of Nudge 
policies $(\texttt{A}_k)_{k=1}^m$ such that $\texttt{A}_1 = \texttt{A}$, $\texttt{A}_m$ is $\gamma$-CounterBoost and the prefactor of $\texttt{A}_k$ is non-increasing in $k$. While it is easy to define a sequence
of $n_{ij}$ functions to transform the $n_{ij}$ functions of $\texttt{A}_1$ into those of $\gamma$-CounterBoost, the problem is
finding a sequence of $n_{ij}$ functions such that each set of $n_{ij}$ functions corresponds to a feasible policy, 
which is necessary to prove that the prefactor reduces when increasing $k$. For this reason a new approach is needed 
to prove optimality of $\gamma$-CounterBoost in a large class of CounterBoost policies when $d > 2$.

\subsection*{Contextual CounterBoost: Solving the Feasibility Problem}

A first contribution in this paper is the introduction of a new class of CounterBoost policies called
{\it Contextual CounterBoost} (CCB) policies. A CCB policy is characterized by a context window size $M$ and a boost function $b$. 
The boost function maps the strings $s$ holding the types of the previous and possibly $M$ next arrivals to a boost value. 
Whenever the server needs to select
a new job for service, it computes the boosted arrival counter for each waiting job by subtracting the boost of the job from its arrival counter.  The job with the smallest boosted arrival counter is served next, where ties are broken in favor of
jobs with a larger boost. Note that $\gamma$-CounterBoost is a CCB policy with a context window of size $M=0$ as the
types of the previous and next jobs do not impact its arrival counter boost.

The main advantage of working with the class of CCB policies as opposed to the set $\mathcal{F}$ of policies characterized
by the $n_{ij}$ functions mentioned before, is that any boost function $b$ corresponds to a feasible policy. In addition,
the class of CCB policies can be shown to be broader than $\mathcal{F}$, 
as for any policy $\texttt{A} \in \mathcal{F}$ one can define
a boost function such that its associated CCB policy schedules the jobs identical to policy $\texttt{A}$.

\subsection*{Main Technical Contribution}
Our main objective exists in proving 
that $\gamma$-CounterBoost minimizes the prefactor within the
class of CCB policies. This means that using a context window size $M$ larger than zero does not offer any
reduction in the prefactor.
To establish this result we show that starting from any CCB policy with boost 
function $b^{(1)}$, we can find a 
sequence of boost functions $(b^{(k)})_{k=1}^m$ with $b^{(m)}$ the boost function of $\gamma$-CounterBoost such that the prefactor of the CCB policy associated with boost function $b^{(k)}$ is non-increasing in $k$.

In order to make sure that the prefactor does not worsen when increasing $k$, it is key to make the difference between
two boost functions $b^{(k)}$ and $b^{(k+1)}$ as small as possible. The Single Swap Theorem that we prove in Section \ref{sec:swaptheorem}
gives a simple condition on when replacing a boost function $b$ by $b'$ improves the prefactor provided that 
the only difference between $b$ and $b'$ is that a type-$i$ job that sees a specific context  will pass a type-$j$
job that also sees a specific (possibly overlapping) context when using $b'$, while it did not using $b$.
The next step, which is the more challenging one, exists in showing that one can find the required sequence
of boost functions to transform any CCB policy in $\gamma$-CounterBoost such that we can
apply our Single Swap Theorem to show that the prefactor is non-increasing.

The key technical novelty required to find this sequence of boost functions lies in replacing the boost function $b$ of a
CCB policy by some other boost function $\hat{b}$ such that the associated CCB policy is still the same, but
the boost values $\hat{b}(s)$ are all independent of each other over the field of the rational numbers. 
Further, as adding a constant to all the boost values does not alter a CCB policy, we can increase all the
$b_i^*$ values by a sufficiently large $\Delta$ such that we may assume that the initial boost values of $b^{(1)}(s)$
are all below the boosts of $\gamma$-CounterBoost. The sequence of boosts is then constructed such that all the boosts
$b(s)$ simultaneously increase towards their target. The independence over the rationals makes sure that
each time that the CCB policy changes, it is a change for which we can show that it does not worsen the prefactor
due to the Single Swap Theorem.

\subsection*{Practical Issues when Implementing $\gamma$-CounterBoost}

The boost function $b^*$ of $\gamma$-CounterBoost requires knowing the values of $\tS{i}$,
for $i=1,\ldots, d$, which may not
be easy to estimate accurately in practice. We therefore include a heavy traffic approximation $b^{heavy}$
for the boost function $b^*$ in Section \ref{sec:heavytraffic} that only requires estimating 
(i) the arrival rate, (ii) the mean job sizes per type, and (iii) the second moment of the overall job size distribution.
We numerically show that the boost function $b^{heavy}$ has a prefactor that is fairly close 
to the prefactor of $\gamma$-CounterBoost and mathematically prove that the prefactors coincide as
$\lambda$ tends to one. This result is a generalization of the result in \cite{charlet2024nudgeM} for $d=2$ types and our
proof for $d \geq 2$ is a natural generalization.

\subsection*{Paper Organization}
The paper is organized as follows. The model under consideration and the class of CounterBoost policies are introduced in Section \ref{sec:model}. The class of Contextual CounterBoost policies is introduced in Section \ref{sec:CCB}, while Section \ref{sec:swaptheorem} and Section \ref{sec:opt} prove the Single Swap Theorem and the optimality of $\gamma$-CounterBoost
within the class of CCB policies, respectively. In Section \ref{sec:heavytraffic} we present a heavy traffic approximation for the optimal boost $b^*$ and show that $\gamma$-CounterBoost and $\gamma$-Boost have the same prefactor as the load tends to one. Numerical results are presented in Section \ref{sec:num} and conclusions
are drawn in Section \ref{sec:conc}.

\section{System Model and CounterBoost Policies}\label{sec:model}
We consider an M/G/1 queue with arrival rate $\lambda$ and a class I job size distribution $S$, 
which includes most well-behaved light-tailed distributions \cite{abate1}, such as all phase-type
distributions \cite{latouche1,HarcholBalter2013} or distributions with a finite support. The mean job size is assumed to be one, such that
the load of the queue equals $\lambda$. For such an M/G/1 queue the response time $R_{FCFS}$ under FCFS scheduling
is known to decay exponentially:
\[ P[R_{FCFS}> t]  \approx c_{FCFS} e^{-\gamma t},\]
where $\gamma$ is the decay rate and $c_{FCFS}$ the FCFS prefactor. Let $\tilde S(-\dr)$ be the Laplace transform
of the job size distribution $S$ evaluated in $-\dr$.

Jobs are assumed to be partitioned into $d$ types and a random job is of
type $k$ with probability $p_k$, where $\sum_{k=1}^d p_k=1$. Job types of consecutive jobs are independent.
The size of a type-$k$ job is denoted as $S_k$ and its Laplace transform evaluated in $-\dr$ is denoted
as $\tilde S_k(-\dr)$. Note that $\tilde S(-\dr) = \sum_{k=1}^d p_k \tilde S_k(-\dr)$. For further use define
for $j=1,\ldots,d+1$:
\[ \tilde T_j(-\dr) = \sum_{k=1}^d p_k \tilde S_k(-\dr)^{1(k< j)},\]
where $1(A)=1$ if $A$ is true and $1(A)=0$ otherwise. Note that $\tilde T_1(-\dr)=1$ and $\tilde T_{d+1}(-\dr)=\tilde S(-\dr)$. 

\begin{definition}
    A CounterBoost policy $\pi$ with boosts $(b_i)_{i=1}^d$ is a policy that schedules jobs in order of increasing boosted arrival counter, which is the arrival counter of the job minus its boost $b_i$ for a type-$i$ job. Ties are broken in favor of the job with the larger boost. Define $N_b(i, j) = (\lfloor b_i - b_j \rfloor)^+$ as the number of previous arrivals in which an arriving type-$i$ job can pass type-$j$ jobs. In other words, when a type-$i$ job arrives, it is served before any type-$j$ job that arrived in the previous $N_b(i, j)$ arrivals, and is still waiting in the queue when the type-$i$ job arrives. The behavior of a CounterBoost policy is fully determined by the function $N_b$, and many boosts exist that produce the same policy. 
\end{definition}

In particular we focus on one policy, called $\gamma$-CounterBoost, for which we show that it is tail optimal within a large family of policies.

\begin{definition}
    $\gamma$-CounterBoost is a CounterBoost policy where a type-$i$ job receives a boost of 
    $$b_i^* = \frac{1}{\log \tS{}} \log \left(\frac{\tS{i}}{\tS{i}-1}\right).$$
    We denote $N^*(i, j) = N_{b^*}(i, j) = (\lfloor b^*_i - b^*_j \rfloor)^+$ for later use.
\end{definition}

When calculating the tail constant of a policy, we will often use the final value theorem and the 
properties of the Laplace transform. Assume $Y=Z + X$ with
$Z$ and $X$ independent, where $Y$ and $Z$
decay exponentially at the same rate $\dr$, meaning $P[Y > t] \sim c_Y e^{-\dr t}$ and
$P[Z > t] \sim c_Z e^{-\dr t}$ for some constant $c_Y$ and $c_Z$. Let $\tilde{X}$, $\tilde{Y}$, and $\tilde{Z}$ denote the Laplace transforms of $X$, $Y$, and $Z$ respectively. Then,
\begin{align*}
    c_Y &= \lim_{t \rightarrow \infty} e^{\dr t} P[Y > t] =
 \lim_{s \rightarrow 0} \frac{1}{\dr} s\tilde Y(s-\dr) 
 = \lim_{s \rightarrow 0} \frac{1}{\dr} s\tilde Z(s-\dr)\tilde X(s-\dr) = c_Z \tilde X(-\dr).
\end{align*}
This shows that adding work to a job's waiting time results in its constant being multiplied by the Laplace transform of the amount of work $X$ that is added. When a job is passed by another job, such a factor is added to the passed job's constant. We also use this property to condition the constant on previous arrivals. For example, we define $c_Z$ as the constant of the workload seen by a random job arrival. If we know that the job before a tagged job is of type $i$, then the prefactor of
the waiting time of the tagged job equals $c_Z \tS{i}/\tS{}$. This can be viewed as removing a random job and then adding a type-$i$ job.

\begin{theorem}\label{th:cboostcontextless}
    Let $c_{W^{(i)}}(b)$ be the prefactor of the waiting time of a type-$i$ job under a CounterBoost policy with boosts $b$. If we relabel the job types such that the boosts satisfy $b_1 \geq b_2 \geq \ldots \geq b_d$, then
    \begin{align}
        c_{W^{(i)}}(b) &= c_Z \prod_{j=2}^{i} \left(\tilde T_j(-\dr)\right)^{N_b(j-1, i) - N_b(j, i)} \prod_{j=i+1}^{d} \left(\frac{\tilde T_j(-\dr)}{\tS{}}\right)^{N_b(i, j) - N_b(i, j-1)} \label{eq:cWicontextless}
    \end{align}
    for $i=1,\ldots,d$, where $c_Z = c_{FCFS}/\tS{}$.
\end{theorem}
\begin{proof}
 The proof makes repeated use of the final value theorem.
    Due to the class I job size distribution, we may assume that there are enough jobs present in the queue such that all the jobs that arrived among the last 
    $N_b(1, d)$ are still waiting in the queue when considering $c_{W^{(i)}}(b)$. This can be argued in the same way as in the proofs for tail constants in \cite{charlet2024nudgeM}.
    For $c_{W^{(1)}}$ we consider the FCFS workload without the last $N_b(1, d)$ arrivals, which contributes a factor $c_Z/\tS{}^{N_b(1, d)}$ to $c_{W^{(1)}}$. We stress that this is not removing the last $N_b(1, d)$ jobs in the queue (whose types are reordered by the policy), but the last $N_b(1, d)$ arrivals, which have the same type distribution for any policy.
    For the $N_b(1, 2)$ most recent arrivals, we only need to add the work of the type-$1$ arrivals, 
    as $N_b(1, 2) \leq N_b(1, k)$ for $k > 2$, which yields a factor
    \[
    \left( p_1 \tS{1} + (1-p_1) \right)^{N_b(1, 2)} = \tT{2}^{N_b(1,2)}.
    \]
    For the next $N_b(1, 3)-N_b(1, 2)$ arrivals we need to add the work of both type-$1$ and type-$2$ arrivals,
    which results in a factor
    \[
    \left( p_1 \tS{1} + p_2 \tS{2} + (1-p_1-p_2) \right)^{N_b(1, 3)-N_b(1, 2)} = \tT{3}^{N_b(1,3)-N_b(1,2)}.
    \]
    In general for the arrivals in position $N_b(1, j-1)+1$ to $N_b(1, j)$, we add the work of the type-$1$ to $j-1$
    jobs and this creates a factor
    \[
    \left( \sum_{k=1}^{j-1} p_k \tS{k} + \sum_{k=j}^d p_k \right)^{N_b(1, j) - N_b(1, j-1)} = \tT{j}^{N_b(1,j)-N_b(1,j-1)}.
    \]
    Combining these factors (for $j=2$ to $d$) with $c_Z/\tS{}^{N_b(1, d)}$ shows that \eqref{eq:cWicontextless} holds for $i=1$.
    
    The prefactor $c_{W^{(d)}}$ is equal to $c_Z$ multiplied by the factors that represent the work
    that passes a type-$d$ job. Any type-$k$ $< d$ arrival in the next $N_b(d-1, d)$ arrivals passes the type-$d$ job
    as $N_b(k, d) \geq N_b(d-1, d)$ for $k < d$,
    yielding a factor
    \[ \left( \sum_{k=1}^{d-1} p_k \tS{k} + p_d \right)^{N_b(d-1, d)} = \tT{d}^{N_b(d-1, d)}.
    \]
    The type-$k$ $< d-1$ jobs in the next $N_b(d-2, d) - N_b(d-1, d)$ arrivals also pass the type-$d$ job,
    which gives rise to a factor
    \[ \left( \sum_{k=1}^{d-2} p_k \tS{k} + p_{d-1} + p_d \right)^{N_b(d-2, d) - N_b(d-1, d)}
    = \tT{d-1}^{N_b(d-2, d) - N_b(d-1, d)}.
    \]
    In general all the type-$k$  $ < j$ jobs in a set of $N_b(j-1, d) - N_b(j, d)$ arrivals pass the type-$d$ job and
    this implies that
    \begin{align*}
        c_{W^{(d)}}(b) &= c_Z  \prod_{j=2}^{d} \tT{j}^{N_b(j-1, d) - N_b(j, d)}. 
    \end{align*}
    Hence, \eqref{eq:cWicontextless} holds for $i=d$.
    For type-$i$ jobs with $1 < i < d$, we get a combination of the above: some type $1$ to $i-1$ pass (these are among the
    next $N_b(1, i)$ arrivals) and some type $i+1$ to $d$ jobs are passed (belonging to the last $N_b(i, d)$ arrivals).
    Reasoning in the same way as above yields \eqref{eq:cWicontextless}.
\end{proof}

Using this theorem, the prefactor $c_{\gamma-CounterBoost}$ of $\gamma$-CounterBoost can be expressed as
\[c_{\gamma-CounterBoost}=\sum_{i=1}^d p_i c_{W^{(i)}}(b^*) \tS{i}.\]

\section{Contextual CounterBoost policies: definition and results}\label{sec:CCB}

We now introduce the class of Contextual CounterBoost policies. In such a policy
the counter boost of a job may not only depend on its own job type, but also on the types of the previous $M$ and next $M$
arrivals. As some of the next $M$ job arrivals may still need to occur when the scheduler needs to select a job for service,
we also specify the counter boost of a job when only $\ell \in \{0,\dots,M-1\}$ of the next $M$ job arrivals have occurred.
These have however no impact on the prefactor of the policy.

\begin{definition}
    A Contextual CounterBoost (CCB) policy $\pi$ with context window $M \in \mathbb{Z}^+$ is characterized by a boost function $b$ which takes as arguments the current arrival type, the previous $M$ arrival types, and $0 \leq \ell \leq M$ future arrival types as context. Denote the set $\{1,\ldots,d\}$ as $[d]$. These three arguments are specified by a string $s \in {[d]}^{M+1+\ell}$ where $s_1$ is the type of the most recent arrival, $s_{\ell+1}$ is the type of the job receiving the boost, and $s_{M+1+\ell}$ is the type of the oldest arrival
    among these $M+1+\ell$ arrivals: 
    \begin{align}
        b: \bigcup_{\ell=0}^M {[d]}^{M+1+\ell} \rightarrow \mathbb{R}. 
    \end{align}
    When the server starts serving a new job, it looks at the boosted arrival counter of all jobs waiting in the queue, which is equal to their arrival counter minus their boost. This boost is given by the boost function, determined by the job type and the previous $M$ types, as well as up to $M$ future types, depending on how many jobs have arrived since. As before, jobs are served in order of boosted arrival counter, and in case multiple jobs have the same boosted counter, ties are broken by giving priority to more recent arrivals (or equivalently, arrivals with a larger boost).
\end{definition}

\begin{definition}
    The class of Contextual CounterBoost policies with context window $M$ is denoted as $\mathcal{G}_M$, while the "Contextless" CounterBoost policies are $\mathcal{G}_0$.
\end{definition}

As the set of strings on which the boost function $b$ acts is finite, we can define $U = \lfloor \max_s b(s)-\min_s b(s) \rfloor = 
\max_{s,s'} N_b(s,s')$, where
$N_b(s,s')= (\lfloor b(s)-b(s')\rfloor)^+$. 
We further define the operator $\crop{k}$ that maps a string $ v=v_{-L_1} \ldots v_{L_2}$ of length $L=L_1+1+L_2 \geq 2M+1$ to the length $2M+1$ substring $v_{k-M}\ldots v_{k+M}$ with $M-L_{1} \leq k \leq L_2-M$.

\begin{theorem}\label{th:boostcontext}
    Let $c_{W^{(i)}}(b)$ be the prefactor of the waiting time of a type-$i$ job under a Contextual CounterBoost policy $\pi_{M}$ with boost function $b$, then:
      \begin{align}
        c_{W^{(i)}}(b) = \frac{c_Z}{\tS{}^{M+U}} &
         \sum_{\substack{v \in {[d]}^{2M+2U+1} \\ v_{0} = i}} \frac{p(v)}{p_i}  \prod_{k=1}^{U} \tS{v_k}^{1(N_b(s, \crop{k} v) < k)}  \nonumber \\
        &\hspace{-1cm}\cdot \prod_{\ell=1}^{U} \tS{v_{-\ell}}^{1 - 1(N_b(\crop{(-\ell)} v, s) < \ell)} \cdot \prod_{m=U+1}^{U+M} \tS{v_m} \label{eq:cWi},
    \end{align}
    with $v=v_{-(U+M)}\ldots v_{U+M}$, $p(v) = \prod_{k=-(U+M)}^{U+M} p_{v_k}$ and $s=\crop{0} v$ is the context of the tagged type-$i$ job with index $0$.
    The prefactor $c_{CCB}(b)$ is given by
    \begin{align}
        c_{CCB}(b) = \sum_{i=1}^d p_i c_{W^{(i)}}(b) \tS{i}.
    \end{align}
\end{theorem}
\begin{proof}
We first remove the work associated to the $M+U$ jobs that arrived before the tagged type-$i$ job from $c_Z$
using the final value theorem.
    The formula for $c_{W^{(i)}}(b)$ then sums over all possible strings $v$ of length $2M+2U+1$ with a tagged type-$i$ job in the center. 
    If we label the entries of $v$ as $v_{-(U+M)}\ldots v_{U+M}$, then the tagged
    type-$i$ job in the middle has index $0$. As the number of jobs that a job can pass is bounded by $U$ and the
    number of jobs that can pass a particular job is also bounded by $U$, it is not hard to prove that we may assume that 
    at least $M+U$ jobs are waiting in the queue when the tagged job arrived, and at least $M+U$ jobs arrive before any of these
    jobs starts service. 
    On both sides of the tagged type-$i$ job, we need to consider the types of $U$ jobs as those might pass the tagged job or be passed by the tagged type-$i$ job. In addition, an extra $M$ jobs need to be considered to provide context for these $U$ jobs' boost function. For each string  $v=v_{-(U+M)}\ldots v_{U+M}$, we first take the probability of this string, conditioned on a type-$i$ job in the middle. 
    Then three products add work back using the final value theorem:
    \begin{enumerate}
        \item The first product adds jobs back that are in the $U$ arrivals before the tagged type-$i$ job,  
        if the tagged job does not pass them.
        These jobs 
        have index $k \in \{1,\ldots,U\}$. The tagged type-$i$ job has context $s$ and index $0$,
        therefore the job with index $k$ and context $\crop{k} v$ is not passed by the tagged type-$i$ job if and only if  $b(s)- b(\crop{k} v) < k$. This condition can be restated as $N_b(s, \crop{k} v) < k$.
        \item The second product adds jobs that are part of the  $U$ arrivals after the tagged type-$i$ job, if they pass the tagged job. These jobs have index $-\ell$ with $\ell \in \{1,\ldots,U\}$. 
        The job with index $-\ell$ passes the tagged type-$i$ job if and only if
        $b(\crop{(-\ell)} v)-b(s) \geq \ell$, which means it does not pass the tagged type-$i$ job if and only
        if $N_b(\crop{(-\ell)} v,s) < \ell$.
        \item The third product adds the work of the $M$ jobs that arrived 
        more than $U$ jobs before the tagged type-$i$ job. These jobs 
        have index $m \in \{U+1,\ldots,U+M\}$ and are always added as the tagged job can not pass them.
    \end{enumerate}
    The constant $c_{CCB}(b)$ is a weighted average of $c_{R^{(i)}}(b) = c_{W^{(i)}}(b) \tS{i}$.
\end{proof}

Note that this constant is only affected by values of the boost function for strings $s$ of length $2M+1$, i.e. the boosts of jobs when fewer than $M$ future arrivals have occurred do not affect the constant $c_{CCB}(b)$.

\section{The Single Swap Theorem}\label{sec:swaptheorem}

In this section we prove the Single Swap Theorem, that is applied repeatedly in the next section to prove tail optimality of the $\gamma$-CounterBoost policy among all CCB policies:

\begin{theorem}[Single Swap Theorem] \label{th:swap}
    Consider a CCB policy that is modified such that for some context $s$ with $s_{M+1} = i$ and  some other context $s'$ with $s'_{M+1} = j$, $N_b(s,s')= (\lfloor b(s) - b(s') \rfloor)^+$ increases from $k-1$ to $k \geq 1$, while all other $N_b$-values remain the same. In other words, a type-$j$ job with context $s'$ that arrived $k$ jobs before a type-$i$ job with context $s$ was initially not passed, but it is after the modification. If this  modification changes the prefactor, then it reduces the prefactor if and only if $k \leq N^*(i,j)$.
\end{theorem}

\begin{proof}
If $k \leq 2M$ in the above theorem, then the last entries of $s$ must equal the first few entries of $s'$ (e.g., 
if $M=2$ and $k=3$, then the last two entries of $s$ match the first two of $s'$), otherwise the policy
does not change.

To prove our result we consider strings $v$ of length $2M+2U+k+1 = M + U + 1 + (k-1) + 1 + U + M$, where the two ones correspond to the tagged type-$i$ and type-$j$ job, the $U$ terms correspond to jobs that can pass or be passed by (at least) one of the tagged jobs, the $k-1$ are the jobs between the tagged jobs and the two $M$ values provide the necessary context.
Let $v = v_{-(U+M)} \ldots v_{U+M+k}$ be a string of length $2M+2U+1+k$ for $k \geq 1$. Let the tagged type-$i$ job with index $0$ have arrival counter 0 without loss of generality, meaning the tagged type-$j$ job with index $k$ has arrival counter $-k$. 

Define $D_{v,0}$  as the amount of work the jobs in $v$ \textit{excluding the tagged type-$j$ job with arrival counter $-k$} 
contribute to the waiting time of the tagged type-$i$ arrival in $v$ with arrival counter $0$.
Similarly, define $D_{v,k}$  as the amount of work the jobs in $v$ \textit{excluding the type-$i$ tagged job with arrival counter $0$} 
contribute to the waiting time of the tagged type-$j$ arrival in $v$ with arrival counter $-k$.

Note that the boost of the first and last $M$ jobs  in a string $v$ is unknown as the context of these $2M$ jobs is not fully included in $v$. This poses no problem as these jobs are only used to provide the context to the intermediate jobs of $v$.
Due to the definition of $U$,  jobs in the first $M$ positions of $v$ never add work to a tagged job with arrival counter in $[-k,0]$. Similarly the $M$ jobs in the last $M$ positions always add work to a tagged job with an arrival counter in $[-k,0]$.
Any job in between (these are jobs with index $\ell \in \{-U,\ldots,k+1+U\} \setminus \{0, k\}$) contributes to $D_{v,0}$ if $-b(\mathcal{W}_0 v)$ is larger than that job's boosted arrival counter given by $(-\ell) - b(\crop{\ell}v)$ and to
$D_{v,k}$ if $(-k)-b(\mathcal{W}_k v)$ is larger than that job's boosted arrival counter given by $(-\ell) - b(\crop{\ell}v)$. 

Using these observations, we now present two alternate representations for \eqref{eq:cWi} using strings of length $2M+2U+k+1$, that is, 
the waiting time constants for the tagged type-$i$ job and tagged type-$j$ job respectively are given by
    \begin{align}
        c_{W^{(i)}}(b) &= \sum_{\substack{v \in [d]^{2M+2U+k+1} \\ v_0 = i}} \frac{p(v)}{p_i} \frac{c_Z}{\tS{}^{M+U+k}} \wtilde{D_{v,0}}(-\dr) \cdot \tS{v_k}^{1[-b(\crop{0} v) > -k-b(\crop{k} v)]},\\
        c_{W^{(j)}}(b) &= \sum_{\substack{v \in [d]^{2M+2U+k+1} \\ v_k = j}} \frac{p(v)}{p_j} \frac{c_Z}{\tS{}^{M+U}} \wtilde{D_{v,k}}(-\dr) \cdot \tS{v_0}^{1[-k-b(\crop{k} v) \geq -b(\crop{0} v)]},
    \end{align}
where $\wtilde{D_{v,0}}(s)$ and $\wtilde{D_{v,k}}(s)$ are the Laplace transform of $D_{v,0}$ and $D_{v,k}$, respectively. 
The reasoning is nearly identical as the one used to prove Theorem \ref{th:boostcontext}, the main change is that we sum over strings that are longer than before. The extra factor at the end of each formula is to take into account the other tagged job, which is not included in $D_{v,0}$ or $D_{v,k}$. This formula makes use of the final value theorem in the same manner as before.

An important observation is as follows: if the only change made to the policy is that the tagged type-$i$ job now passes the tagged type-$j$ job, then this must mean that both before and after the change, the two jobs were served consecutively, with only their order switching around. If a third job were served after the type-$j$ job but before the type-$i$ job before the change, this third job would either (a) also be passed by the type-$i$ job after the change (but not before), or (b) would pass the type-$j$ job after the change (but not before), in each case this would be more than the one change we allow to the policy. This lets us state that for any string $v$ where the two tagged jobs have contexts $s$ and $s'$ respectively, $D_{v,0} = D_{v,k}$ both before and after the change, which we simply denote as $D_v$.

We now define the set of all such strings $v$ given that the type-$i$ job with arrival counter $0$ has context $s$ and the type-$j$ job with arrival counter $-k$ has context $s'$.
For $s, s' \in {[d]}^{2M+1}$ and $k \geq 1$, define 
\begin{align*}
    V(s, s', k) &= \left\{ v \in {[d]}^{2M+2U+1+k} \hspace{2mm} | \hspace{2mm} \crop{0} v = s, \hspace{1mm} \crop{k} v = s' \right\},
\end{align*}
as the set of strings of length $2M+2U+1+k$ with $s$ and $s'$ in $v$, centered around index $0$ and $k$ respectively.

We can now look at how the change to the policy affects $c_{W^{(i)}}(b)$ and $c_{W^{(j)}}(b)$. We add a $\Delta$ to the notation whenever we refer to the change in some variable, and start out by only considering the terms of the summation 
for $c_{W^{(i)}}(b)$ and $c_{W^{(j)}}(b)$ that are actually affected by the modification to the policy. For strings $v \notin V(s, s', k)$, the amount of work the tagged jobs wait for does not change, while for strings in $V(s, s', k)$, we use the aforementioned $D_v$ notation as the two tagged jobs wait for the same amount of work in $v$. The only difference is that the type-$i$ job is now served before the type-$j$ job, which explains the
two factors at the end of the following expressions:

\begin{align}
    \Delta c_{W^{(i)}} &= \sum_{v \in V(s, s', k)} \frac{p(v)}{p_{i}} \frac{c_Z}{\tS{}^{M+U+k}} \wtilde{D_v}(-\dr) (1 - \tS{j}), \\
    \Delta c_{W^{(j)}} &= \sum_{v \in V(s, s', k)} \frac{p(v)}{p_{j}} \frac{c_Z}{\tS{}^{M+U}} \wtilde{D_v}(-\dr) (\tS{i} - 1).
\end{align}

If $V(s, s', k)$ is empty, then the policy stays the same and therefore $c_{CCB}(b)$ remains unchanged.
Otherwise, this change gives a reduction in $c_{CCB}(b)$ if and only if
\begin{align*}
    p_{i} \tS{i} \Delta c_{W^{(i)}} + p_{j} \tS{j} \Delta c_{W^{(j)}} < 0,
\end{align*}
meaning
\begin{align*}  
    \left( \sum_{v \in V(s, s', k)} p(v) \frac{c_Z \wtilde{D_v}(-\dr)}{\tS{}^{M+U}} \right) \left(\frac{\tS{i}}{\tS{}^k} (1 - \tS{j}) + \tS{j} (\tS{i} - 1) \right) < 0.
\end{align*}

The first factor containing the summation only includes positive terms, therefore we get a reduced prefactor $c_{CCB}(b)$ if and only if
\begin{align*}
    \frac{\tS{i}}{\tS{}^{k}} (1 - \tS{j}) + \tS{j} (\tS{i} - 1) &< 0,
\end{align*}
which is equivalent to
    \begin{align*}
    \tS{}^{k} &< \frac{\tS{i}}{\tS{i} - 1} \frac{\tS{j} - 1}{\tS{j}}.
\end{align*}   
This condition holds if and only if $k \leq N^*(i,j)$.
\end{proof}

\section{Optimality of $\gamma$-CounterBoost}\label{sec:opt}

In this section, we prove that $\gamma$-CounterBoost is tail optimal within $\cup_M \mathcal{G}_{M}$, that is, the class of all CCB policies. 
As the tail behavior is only affected by boosts for strings of length $2M+1$, we can ignore the boosts for strings with a length below $2M+1$.
For this proof we rely on the Single Swap Theorem, proven in Section \ref{sec:swaptheorem}.
The main challenge in using the Single Swap Theorem 
exists in showing that starting from any
CCB policy, we can modify the boosts in several steps such that (a) the CCB policy becomes identical to $\gamma$-CounterBoost, (b)
the prefactor $c_{CCB}(b)$ does not increases during each step and (c) the modification made to the CCB policy during each step is as stated in the requirements of the Single Swap Theorem. To achieve this last requirement, we make use of the fact that we can 
always slightly modify the boost values for some strings $s$ without actually changing the order in which jobs are scheduled.
To give a very simple example. Suppose we have $2$ job types and the boost $b(s)$ only depends on the type $i$ of a job and not its
context. Let $b_1$ be the boost of a type-$1$ job and $b_2$ the boost of a type-$2$ job. In this case we can freely
modify $b_1$ and $b_2$ as long as $N_b(1, 2) = (\lfloor b_1 - b_2 \rfloor)^+$ and $N_b(2, 1) = (\lfloor b_2 - b_1 \rfloor)^+$ remain the same, without modifying the order in which jobs are scheduled.

\begin{lem} \label{lem:ratboosts}
    Consider a CCB policy in $\mathcal{G}_{M}$ with boost function $b$. We can replace the boost function by $\bar{b}$ so that the policy schedules jobs in the same way, while all boosts are rational numbers.
\end{lem}
\begin{proof}
    We define the new boost function as $\bar{b}(s) = round(b(s), n)$, rounding $b(s)$ to $n \geq 1$ decimal digits.
    Let $\{x\}$ be the fractional part of $x \in \mathbb{R}$. For all $s, s'$ where $\{ b(s) - b(s') \} = 0$, this remains the case after rounding for any $n$ and so $\lfloor b(s) - b(s') \rfloor = \lfloor \bar{b}(s) - \bar{b}(s') \rfloor$. Otherwise if $\{ b(s) - b(s') \} \not= 0$, note that $| (\bar{b}(s) - \bar{b}(s')) - (b(s) - b(s')) | \leq 10^{-n}$. If we pick $n$ so that $10^{-n} < 1 - \max_{s, s'}\{ b(s) - b(s') \}$, then for all $s, s'$, $\lfloor b(s) - b(s') \rfloor =  \lfloor \bar{b}(s) - \bar{b}(s') \rfloor$ holds. All floors remaining the same implies $N_b(s, s') = N_{\bar{b}}(s, s')$ for all $s, s'$, meaning the two policies schedule jobs identically.
\end{proof}

We illustrate this on the following example which we use throughout this section. Consider a policy with $d = 4$ types and $M = 0$, i.e. no context. For this example we denote the boost as $b_i$ instead of $b(s)$. Assume the boosts $b$ start out as follows:
\begin{align}
    b_1 &= 2.418 \ldots & b_2 &= 1.950 \ldots & b_3 &= 3.721 \ldots & b_4 &= 4.234 \ldots \nonumber
\end{align}
We pick $n = 2$ which satisfies $10^{-n} < 1 - \max_{s, s'} \{ b(s) - b(s') \}$ and changes the boosts to
\begin{align}
    b_1 &= 2.42 & b_2 &= 1.95 & b_3 &= 3.72 & b_4 &= 4.23 \nonumber
\end{align}
This results in the same policy, now using only rational boosts.

\begin{lem} \label{lem:irratboosts}
    Consider two CCB policies in $\mathcal{G}_{M}$ with boost functions $b$ and $b'$. Then there exist boost functions $\hat{b}$ and $\hat{b'}$ such that
    \begin{enumerate}
        \item $N_{b}(s, s') = N_{\hat{b}}(s, s')$ and $N_{b'}(s, s') = N_{\hat{b'}}(s, s')$ for all $s, s'$ of length $2M+1$, meaning replacing $b$ by $\hat{b}$ or $b'$ by $\hat{b'}$ does not change the policies' prefactors, and
        \item all $d^{2M+1}$ values of both $\hat{b}(s)$ and $\hat{b'}(s)$ for strings of length $2M+1$ are linearly independent over the field of rational numbers, meaning for rational $r_0$, $r(s), r'(s)$, the only solution to $\sum_s r(s) \hat{b}(s) + 
        \sum_s r'(s) \hat{b'}(s) = r_0$ is when $r_0 = r(s) = r'(s) = 0$ for all $s$.
        \item $\hat{b}(s) < \hat{b'}(s)$ for all strings $s$ of length $2M+1$.
    \end{enumerate}
\end{lem}
\begin{proof}
    We start by applying Lemma \ref{lem:ratboosts} to both $b$ and $b'$ to ensure both policies only have rational boosts, with boost functions $\bar{b}$ and $\bar{b'}$ respectively. Let $q_1, \ldots, q_{2d^{2M+1}}$ be the first $2d^{2M+1}$ primes, such that their square roots are independent over the rational numbers. Define $\epsilon_i = (q_i)^{1/2}/\lceil (q_{2d^{2M+1}})^{1/2} \rceil < 1$. 

    The idea is to find one sufficiently small number $\epsilon(s)$ per string $s$ where $\bar{b}(s) > \bar{b}(s')$ implies $\epsilon(s) > \epsilon(s')$. The new boost function $\hat{b}(s) = \bar{b}(s) + \epsilon(s)$ should then result in the same floor values (and thus $N_{\bar{b}}(s, s') = N_{\hat{b}}(s, s')$).

    To make sure that $\lfloor \bar{b}(s) - \bar{b}(s') \rfloor = \lfloor 
    \hat{b}(s) - \hat{b}(s') \rfloor $ when $\bar{b}(s) > \bar{b}(s')$, 
    all epsilons need to be less than $\epsilon_{max} = 1-\max_{\tilde s,\tilde s'} \{ \bar{b}(\tilde s) - \bar{b}(\tilde s') \}$. In this case
    we have
    \begin{align}
        \lfloor \hat{b}(s)-\hat{b}(s') \rfloor = \lfloor \bar{b}(s)-\bar{b}(s') + (\epsilon(s) - \epsilon(s')) \rfloor \geq \lfloor \bar{b}(s)-\bar{b}(s') \rfloor
    \end{align}
    as $\epsilon(s) > \epsilon(s')$ and
    \begin{align}
        \lfloor \hat{b}(s)-\hat{b}(s') \rfloor = \lfloor \bar{b}(s)-\bar{b}(s') + (\epsilon(s) - \epsilon(s')) \rfloor \leq \lfloor \bar{b}(s)-\bar{b}(s') \rfloor
    \end{align}
    as $\epsilon(s) - \epsilon(s') < \epsilon_{max}$, when $\bar{b}(s) > \bar{b}(s')$. 
    
    If $\bar{b}(s) = \bar{b}(s')$, we also assign different epsilons to both. Assume $\epsilon(s) > \epsilon(s')$ without loss of generality.
    In this case $\lfloor \bar{b}(s) - \bar{b}(s') \rfloor = 0 = \lfloor \epsilon(s) - \epsilon(s') \rfloor = \lfloor \hat{b}(s) - \hat{b}(s') \rfloor$. The reverse is however not true: $\lfloor \bar{b}(s') - \bar{b}(s) \rfloor = 0$ and $\lfloor \hat{b}(s') - \hat{b}(s) \rfloor = -1$. Fortunately this does not affect the policy as $N_{\hat{b}}(s', s) = 0$. 
    Finally, in case $\bar{b}(s) < \bar{b}(s')$, then $N_{\bar{b}}(s, s') = 0$ and as $\epsilon(s) < \epsilon(s')$ we also have
    $N_{\hat{b}}(s, s') = 0$.
    
    To find the epsilons, we set $\epsilon(s)$ equal to $\epsilon_{max} \epsilon_{d^{2M+1} + 1 - k}$, where $s$ is the string with the $k$-th largest boost $\bar{b}(s)$, with $1 \leq k \leq d^{2M+1}$. In case multiple strings have the same boost, ties may be broken arbitrarily. This fixes the boosts $\hat{b}(s)$ for all $s$.
    
    We now use the same argument to set the $\hat{b'}(s)$ values. Let $\epsilon'_{max}=1-\max_{\tilde s, \tilde s'} \{\bar{b'}(\tilde s)-\bar{b'}(\tilde s')\}$, then $\hat{b'}(s) = \bar{b'}(s) +  \epsilon'_{max} \epsilon_{2d^{2M+1} + 1 - k}$ where $s$ is the string with the $k$-th largest boost $\bar{b'}(s)$, using the second half of the independent irrational numbers.

    Finally, we increase all boosts $\hat{b'}(s)$ by some fixed rational amount $c > \max_s \hat{b}(s) - \min_s \hat{b'}(s)$, which ensures $\hat{b}(s) < \hat{b'}(s)$ for all strings $s$ as required. Adding a constant to all boosts clearly does not change the policy.
\end{proof}

When using this property in our proof of the next theorem, the CCB policy corresponding to $b'$ (and $\hat b'$) is the
$\gamma$-CounterBoost policy. We pick one more independent irrational number $\epsilon > 0$, for example $\epsilon = (q_{2d^{2M+1} + 1})^{1/2}$. In the following proof, boosts will change throughout the procedure. We would like to represent each boost as a vector of rational coefficients with $2 + 2d^{2M+1}$ dimensions, i.e. a vector in $\mathbb{Q}^{2 + 2d^{2M+1}}$. The vector representing the "bases" is given as 
$$\vec{b} = [1, \epsilon, b(11 \ldots 11), b(11 \ldots 12), \ldots, b(dd \ldots dd), b^*(11 \ldots 11), \ldots, b^*(dd \ldots dd)].$$
The first dimension is for adding rational offsets, the second dimension is to add arbitrarily small epsilons during the procedure, while the remaining dimensions contain the \textbf{initial} boost values, which are all linearly independent as well. We call the dimensions the rational dimension, the $\epsilon$ dimension, and the remaining ones the $s$ or $s^*$ dimension, for strings $s$ of length $2M+1$.
Let $\mathcal{H} = \{ \vec{v} \vec{b} \hspace{2mm} | \hspace{2mm} \vec{v} \in \mathbb{Q}^{2 + 2d^{2M+1}} \}$ be the set of numbers in this vector space which is closed under addition and subtraction, and $\mathbb{Q} \subset \mathcal{H} \subset \mathbb{R}$.
Two numbers $x, y \in \mathcal{H}$ are the same if and only if their vector representation is identical. 
Further, their difference is an integer if and only if the vector representation of their difference contains zeroes in all dimensions, except in the rational dimension where it contains the integer. We will use this to argue during the proof that certain numbers cannot be the same, as their vector representation must be different.

For our running example, suppose that after applying Lemma \ref{lem:ratboosts}, the target boosts $b^*$ are given by
\begin{align}
    b^*_1 &= 3.18 & b^*_2 &= 2.06 & b^*_3 &= 5.75 & b^*_4 &= 3.79 \nonumber
\end{align}
while the starting boosts $b$ are as before. We apply Lemma \ref{lem:irratboosts}: the first eight primes are 2, 3, 5, 7, 11, 13, 17, 19, so that $\epsilon_i = \sqrt{q_i}/5$. We first change $b$ using the first four epsilons. The bound $\epsilon_{max} = 1-\max_{\tilde s,\tilde s'} \{ b(\tilde s) - b(\tilde s') \} = 0.19 = 19/100$ is multiplied by these epsilons which are then added to $b(s)$, with the largest epsilons added to the largest boosts:
\begin{align}
    b_1 &= 2.42+\frac{19}{500} \sqrt{3} \approx 2.486 & b_2 &= 1.95+\frac{19}{500} \sqrt{2} \approx 2.004 \nonumber \\
    b_3 &= 3.72+\frac{19}{500} \sqrt{5} \approx 3.805 & b_4 &= 4.23+\frac{19}{500} \sqrt{7} \approx 4.331 \nonumber
\end{align}
This is repeated on $b^*$ with a different bound $\epsilon^*_{max} = 1-\max_{\tilde s,\tilde s'} \{ b^*(\tilde s) - b^*(\tilde s') \} = 1/25$ and using the last four epsilons. We also add $c = 10 > \max_s b(s) - \min_s b^*(s)$ to $b^*$ to ensure $b(s) < b^*(s)$.
\begin{align}
    b^*_1 &= 3.18+\frac{1}{125} \sqrt{13} + 10 \approx 13.209 & b^*_2 &= 2.06+\frac{1}{125} \sqrt{11} + 10 \approx 12.087 \nonumber \\
    b^*_3 &= 5.75+\frac{1}{125} \sqrt{19} + 10 \approx 15.785 & b^*_4 &= 3.79+\frac{1}{125} \sqrt{17} + 10 \approx 13.823 \nonumber
\end{align}

The extra epsilon $\epsilon = \sqrt{23}$ is simply the square root of the next prime. This gives a base vector of
\begin{align}
    \vec{b} = [1, \epsilon, b_1, b_2, b_3, b_4, b^*_1, b^*_2, b^*_3, b^*_4]. \nonumber
\end{align}
We now present the proof and then apply the first few steps to our running example.

\begin{theorem}
    $\gamma$-CounterBoost has the lowest prefactor among all the policies in $\cup_M \mathcal{G}_{M}$.
\end{theorem}
\begin{proof}
    The proof exists in coming up with a sequence of Contextual CounterBoost policies such that starting from an arbitrary policy, 
    each policy either reduces the prefactor of the previous policy in the sequence or keeps it the same, and the final policy schedules jobs for any context $s$ of length $2M+1$ identical to $\gamma$-CounterBoost by setting 
    $b(s) = b^*(s)= b_{s_{M+1}}^*$. 

    We start with an arbitrary CCB policy with boost function $b$ which we will improve by simultaneously increasing all the boosts that are below $b^*(s)$.
    We first apply Lemma \ref{lem:irratboosts} to change both $b$ and $b^*$ to have independent irrational values without affecting the policy's behavior. Additionally, this ensures $b(s) < b^*(s)$ for all strings $s$. 

    Let $W = {[d]}^{2M+1}$ denote the set of all strings $s$ of length $2M+1$. We partition $W$ into two sets
    \begin{align}
        W^+ = \{ s \in W \hspace{2mm} | \hspace{2mm} b(s) < b^*(s) \} \\
        W^0 = \{ s \in W \hspace{2mm} | \hspace{2mm} b(s) = b^*(s) \}
    \end{align}
    consisting respectively of the strings $s$ for which we are increasing $b(s)$, and the strings for which the boost has reached its target and remains the same.
    For strings in $W^+$, we will increase the boosts by a fixed amount repeatedly until the set becomes empty. This set initially contains all strings, while $W^0$ starts as an empty set.
    
    Define the following two multisets and values:
    \begin{align}
        D_a &= \{ b^*(s) - b(s) \hspace{2mm} | \hspace{2mm} s \in W^+ \} & I_a &= \min D_a \\
        D_c &= \{ \{ b(s) - b(s') \} \hspace{2mm} | \hspace{2mm} s \in W^0, s' \in W^+ \} & I_c &= \min D_c
    \end{align}

    $I_a$ is the distance until a string $s$ sees its boost reach the target value $b^*(s)$, while $I_c$ is the distance before a first pair of floors changes, possibly affecting the policy. For example, suppose $b(s)=4.2$ and $b(s')=3$, then  $0.2
    \in D_c$. Suppose that $I_c = 0.2$, $s' \in W^+$ and $s \in W^0$. Therefore $b(s')$ increases while $b(s)$ remains the same
    during our procedure. When $b(s')$ 
    increases by $0.2$ to $3.2$, $\lfloor b(s) - b(s') \rfloor$ remains the same, but its reverse $\lfloor b(s') - b(s) \rfloor$ changes from $-2$ to $-1$. Increasing $b(s')$ by $\delta$ (which can be chosen arbitrarily small), the $\lfloor b(s) - b(s') \rfloor$ changes from $1$ to $0$.

    We now indicate how the policy changes one step at a time. We refer to several invariant properties labeled (i1) to (i5) which are presented after the steps.
    
    Depending on whether $I_a$ or $I_c$ is smaller, one of the following two cases may occur. Note that $I_a$ cannot be equal to $I_c$ because of invariant (i5).
    \begin{enumerate}
        \item If $I_a < I_c$, then we increase the boosts of all strings in $W^+$ by $I_a$. This causes no $N_b(s, s')$ to change, while for exactly one string in $W^+$, the target boost $b^*(s)$ is reached. This string is then moved from $W^+$ to $W^0$. Because of invariant (i3), it can not be the case that multiple strings reach their target boost at the same time.
        
        \item If $I_c < I_a$, increasing the boosts of all strings in $W^+$ by $I_c$ increases exactly one $\lfloor b(s') - b(s) \rfloor$ from $k-1$ to $k$, for some $k$, with $s \in W^0$ and $s' \in W^+$. 
        The reason only one floor changes is because $D_c$ does not contain any duplicates (invariant (i4)) meaning it has a unique minimum. We then add an extra increment of $q \epsilon$ to all the boosts in $W^+$ which decreases only $\lfloor b(s) - b(s') \rfloor$ from $-k$ to $-(k+1)$, when $q \in \mathbb{Q}^+$ is chosen small enough: $0 < q \epsilon < 1-\max_{\tilde s, \tilde s'} \{ b(\tilde s) - b(\tilde s') \}$, and $0 < q \epsilon < \min_{\tilde s \in W^+} b^*(\tilde s) - b(\tilde s)$ (using the updated boost function which already has $I_c$ added to $W^+$). These conditions are required so that a) the policy does not change because $q \epsilon$ was too large and b) no boost should jump past its target boost.
        
        We split this into three subcases depending on the value of $k$:

        \begin{enumerate}
            \item If $k \geq 1$, then the increase of $\lfloor b(s') - b(s) \rfloor$ changed the policy. We can use Theorem \ref{th:swap} to see that the new policy is either an improvement, or makes no change. As $b(s') \leq b^*(s')$ and $b(s) = b^*(s)$, this implies $k = \lfloor b(s') - b(s) \rfloor \leq \lfloor b^*(s') - b^*(s) \rfloor = N^*(s'_{M+1}, s_{M+1})$ meaning the change is an improvement. The decrease of $\lfloor b(s) - b(s') \rfloor$ from $-k$ to $-k-1$ does not change the policy as $N_b(s, s') = 0$ both before and after changing the boosts.
            \item If $k \leq -1$, the situation reverses and the increase of $\lfloor b(s') - b(s) \rfloor$ does not change the policy, while the decrease of $\lfloor b(s) - b(s') \rfloor$ does. We can again use Theorem \ref{th:swap} to see that undoing this change (that is, increasing the floor back from $-(k+1)$ to $-k$) would make the policy worse, meaning the change improved the policy.
            \item If $k = 0$, then neither floor affects the policy as $N_b(s, s') = N_b(s', s) = 0$ both before and after changing the boosts. This happens when one boost overtakes the other.
        \end{enumerate}
    \end{enumerate}

    Throughout the procedure, there are the following invariants that hold:
    \begin{enumerate}
        \item[(i1)] For each string $s^+ \in W^+$, the $s^+ {}^*$ dimension of $b(s)$ for any $s$ is zero. When the string $s^+$ reaches its target boost $b^*(s^+)$, all strings in $W^+$ now have a one in the $s^+ {}^*$ dimension. This remains at one for $s^+$ as its boost no longer changes. For other strings $s$, the $s^+ {}^*$ dimension may go back and forth between zero and one several times until it eventually stays at zero when these strings reach their own target boost. In particular, the dimension changes to one (or remains one) when a floor that includes $b(s^+)$ changes, while other changes set it to zero again.
        
        \item[(i2)] Initially, the $\epsilon$ dimension is zero in $b(s)$ for each string $s$. Following a step where $I_c < I_a$, this dimension is changed to one rational number $q > 0$ for strings in $W^+$, while it remains zero for strings in $W^0$. In the next step, both $I_a$ and $I_c$ will have $-q$ in this dimension. This means the next step is guaranteed to clear this dimension, and in case $I_c < I_a$, a new (possibly different) $q$ will be added.
        
        \item[(i3)] $D_a$ contains no duplicates: initially, this is clearly the case, as each difference $b^*(s) - b(s)$ has a different vector representation. Each step in the procedure adds the same increment to all boosts for strings in $W^+$, reducing all distances in $D_a$ by the same amount thus introducing no duplicates. When strings reach their target boost, the set $D_a$ becomes smaller.
        
        \item[(i4)] $D_c$ contains no duplicates: the multiset $\{ b(s) - b(s') \hspace{2mm} | \hspace{2mm} 
        s \neq s' \}$ contains no duplicates, contains no integers, and does not contain any pairs of elements the difference 
        of which is an integer. This holds initially as all $b(s)$ have a one in a unique dimension in their vector representation. Similarly to the previous case, each step in the procedure adds an increment to the boosts of all strings in $W^+$. We split this into two cases, depending on whether $I_a$ or $I_c$ was smaller in the step.
        
        If $I_a < I_c$, then the increment was $I_a = b^*(s^+) - b(s^+)$ for an $s^+ \in W^+$. This leaves some $b(s) - b(s')$ the same, increases some by $I_a$, and decreases others by $I_a$. This cannot introduce duplicates as all boosts had a zero in the $s^+ {}^*$ dimension because of invariant (i1), while the $b(s) - b(s')$ that change either have a 1 or a -1 in this dimension. For the same reason, this also introduces no integers. Finally, the pairwise differences originally also had a zero in the $s^+ {}^*$ dimension which now becomes a number $v \in \{-2, -1, 0, 1, 2 \}$. If $v \neq 0$, this clearly cannot be an integer, and when $v = 0$, this pairwise difference which was not an integer remained the same.
        
        If $I_c < I_a$, then the increment was $\{ b(s) - b(s') \} + q \epsilon$ for some $s \in W^0, s' \in W^+, q \in \mathbb{Q}^+$. We split this increment into two parts: by first adding $\{ b(s) - b(s') \}$, we clear the $\epsilon$ dimension as explained in invariant (i2). 
        By then adding $q \epsilon$, we can reuse the same argument as before as the $\epsilon$ dimension is now a "new" dimension, just like the $s^+ {}^*$ dimension in the previous case. Note that if $q = 0$, then an integer would be introduced!
        
        This property implies that $D_c$ does not contain duplicates, as $\{b(s) - b(s')\} = \{b(\bar{s}) - b(\bar{s'})\}$ would imply $(b(s) - b(\bar{s})) + (b(\bar{s'}) - b(s'))$ is an integer.
        
        \item[(i5)] $I_a \neq I_c$, meaning no floors change at the exact moment when a boost $b(s)$ reaches $b^*(s)$. Let $s^+$ be the unique string in $W^+$ such that $I_a = b^*(s^+) - b(s^+)$. As $s^+$ has not reached its target boost yet, vectors $b(s)$ for any $s$ have a zero in the $s^+ {}^*$ dimension due to invariant (i1), while $I_a$ does not have a zero in this dimension.
    \end{enumerate}

    At the end of this procedure, $W^+$ will be empty and all boosts will have reached their target boost. As for all strings $s$
    of length $2M+1$ we have $b(s) = b^*(s)$, the policy has the same prefactor as $\gamma$-CounterBoost. 
    This shows that this policy is tail optimal within $\cup_M \mathcal{G}_{M}$, finishing the proof.
\end{proof}

We now illustrate the first few steps of the procedure on our running example.
The initial vectors representing $b(s)$ and $b^*(s)$ are as follows, where the columns are labeled to indicate the 10 dimensions are those of $\vec{b}$:

\begin{center}
    \begin{tabular}{|cc|cccc|cccc||l|}
        \hline $1$ & $\epsilon$ & $b_1$ & $b_2$ & $b_3$ & $b_4$ & $b^*_1$ & $b^*_2$ & $b^*_3$ & $b^*_4$ & \\ \hline   
        \hline $0$ & $0$ & $1$ & $0$ & $0$ & $0$ & $0$ & $0$ & $0$ & $0$ & $b_1 \approx 2.486$ \\
        \cline{1-10} $0$ & $0$ & $0$ & $1$ & $0$ & $0$ & $0$ & $0$ & $0$ & $0$ & $b_2 \approx 2.004$ \\
        \cline{1-10} $0$ & $0$ & $0$ & $0$ & $1$ & $0$ & $0$ & $0$ & $0$ & $0$ & $b_3 \approx 3.805$ \\
        \cline{1-10} $0$ & $0$ & $0$ & $0$ & $0$ & $1$ & $0$ & $0$ & $0$ & $0$ & $b_4 \approx 4.331$ \\
        \cline{1-10} $0$ & $0$ & $0$ & $0$ & $0$ & $0$ & $1$ & $0$ & $0$ & $0$ & $b^*_1 \approx 13.209$ \\
        \cline{1-10} $0$ & $0$ & $0$ & $0$ & $0$ & $0$ & $0$ & $1$ & $0$ & $0$ & $b^*_2 \approx 12.087$ \\
        \cline{1-10} $0$ & $0$ & $0$ & $0$ & $0$ & $0$ & $0$ & $0$ & $1$ & $0$ & $b^*_3 \approx 15.785$ \\
        \cline{1-10} $0$ & $0$ & $0$ & $0$ & $0$ & $0$ & $0$ & $0$ & $0$ & $1$ & $b^*_4 \approx 13.823$ \\
        \hline
    \end{tabular}
\end{center}

At the start of the procedure we have $W^+ = \{1, 2, 3, 4\}$ and $W^0 = \emptyset$. This means all boosts will be increased at the same time, and so no floors can currently change. Consequently, $D_c$ is an empty set, while $D_a$ has the following four values:

\begin{center}
    \begin{tabular}{|cc|cccc|cccc||l|}  
        \hline $0$ & $0$ & $-1$ & $0$ & $0$ & $0$ & $1$ & $0$ & $0$ & $0$ & $b^*_1-b_1 \approx 10.723$ \\
        \cline{1-10} $0$ & $0$ & $0$ & $-1$ & $0$ & $0$ & $0$ & $1$ & $0$ & $0$ & $b^*_2-b_2 \approx 10.083$ \\
        \cline{1-10} $0$ & $0$ & $0$ & $0$ & $-1$ & $0$ & $0$ & $0$ & $1$ & $0$ & $b^*_3-b_3 \approx 11.980$ \\
        \cline{1-10} $0$ & $0$ & $0$ & $0$ & $0$ & $-1$ & $0$ & $0$ & $0$ & $1$ & $b^*_4-b_4 \approx 9.492$ \\
        \hline
    \end{tabular}
\end{center}

The lowest value is $b^*_4 - b_4 \approx 9.492$: when we increase all $b(s)$ by this amount, the policy does not change, while for one out of four strings, the target boost is reached. The set $W^+$ becomes $\{1, 2, 3\}$ while $W^0$ changes to $\{4\}$. The updated boosts, $D_a$, and $D_c$ are shown below. For clarity, we mark the values that changed in bold.

\begin{center}
    \begin{tabular}{|cc|ccc>{\boldmath}c|ccc>{\boldmath}c||l|}
        \hline $0$ & $0$ & $1$ & $0$ & $0$ & $-1$ & $0$ & $0$ & $0$ & $1$ & $b_1 \approx 11.978$ \\
        \cline{1-10} $0$ & $0$ & $0$ & $1$ & $0$ & $-1$ & $0$ & $0$ & $0$ & $1$ & $b_2 \approx 11.496$ \\
        \cline{1-10} $0$ & $0$ & $0$ & $0$ & $1$ & $-1$ & $0$ & $0$ & $0$ & $1$ & $b_3 \approx 13.297$ \\
        \cline{1-10} $0$ & $0$ & $0$ & $0$ & $0$ & $0$ & $0$ & $0$ & $0$ & $1$ & $b_4 \approx 13.823$ \\
        \hline
        \hline $0$ & $0$ & $-1$ & $0$ & $0$ & $1$ & $1$ & $0$ & $0$ & $-1$ & $b^*_1-b_1 \approx 1.231$ \\
        \cline{1-10} $0$ & $0$ & $0$ & $-1$ & $0$ & $1$ & $0$ & $1$ & $0$ & $-1$ & $b^*_2-b_2 \approx 0.590$ \\
        \cline{1-10} $0$ & $0$ & $0$ & $0$ & $-1$ & $1$ & $0$ & $0$ & $1$ & $-1$ & $b^*_3-b_3 \approx 2.487$ \\
        \hline
        \hline $-1$ & $0$ & $-1$ & $0$ & $0$ & $1$ & $0$ & $0$ & $0$ & $0$ & $\{b_4-b_1\} \approx 0.845$ \\
        \cline{1-10} $-2$ & $0$ & $0$ & $-1$ & $0$ & $1$ & $0$ & $0$ & $0$ & $0$ & $\{b_4-b_2\} \approx 0.327$ \\
        \cline{1-10} $0$ & $0$ & $0$ & $0$ & $-1$ & $1$ & $0$ & $0$ & $0$ & $0$ & $\{b_4-b_3\} \approx 0.526$ \\
        \hline
    \end{tabular}
\end{center}

We see that as string $4$ reaches its target boost, the $b^*_4$ column becomes $1$ for $b_4$, but also for all boosts that were increasing (which are the three remaining boosts in this case), as explained in invariant (i1).

The smallest value is given by $\{b_4 - b_2\}$. When $b_1, b_2, b_3$ are increased by this amount, $\lfloor b_2 - b_4 \rfloor$ increases from -3 to -2, and when a small extra epsilon $q \epsilon$ is added, $\lfloor b_4 - b_2 \rfloor$ decreases from 2 to 1. The former does not affect the policy, but the latter does. For this we use Theorem \ref{th:swap} in reverse: as $k = 2 > N^*(4, 2) = 1$, increasing $\lfloor b_4 - b_2 \rfloor$ from 1 to 2 would make the policy worse according to Theorem \ref{th:swap}, which means decreasing it from 2 to 1 (as we just did) improves the policy.
We pick $q = 10^{-3}$ which makes $q \epsilon \approx 0.00480$ small enough. After this step, the updated boosts, $D_a$, and $D_c$ are as follows:

\begin{center}
    \begin{tabular}{|c>{\boldmath}c|c>{\boldmath}cc>{\boldmath}c|cccc||l|}
        \hline $-2$ & $10^{-3}$ & $1$ & $-1$ & $0$ & $0$ & $0$ & $0$ & $0$ & $1$ & $b_1 \approx 12.310$ \\
        \cline{1-10} $-2$ & $10^{-3}$ & $0$ & $0$ & $0$ & $0$ & $0$ & $0$ & $0$ & $1$ & $b_2 \approx 11.828$ \\
        \cline{1-10} $-2$ & $10^{-3}$ & $0$ & $-1$ & $1$ & $0$ & $0$ & $0$ & $0$ & $1$ & $b_3 \approx 13.629$ \\
        \cline{1-10} \unboldmath{$0$} & \unboldmath{$0$} & \unboldmath{$0$} & \unboldmath{$0$} & \unboldmath{$0$} & \unboldmath{$0$} & \unboldmath{$0$} & \unboldmath{$0$} & \unboldmath{$0$} & \unboldmath{$1$} & $b_4 \approx 13.823$ \\
        \hline
        \hline $2$ & $-10^{-3}$ & $-1$ & $1$ & $0$ & $0$ & $1$ & $0$ & $0$ & $-1$ & $b^*_1-b_1 \approx 0.899$ \\
        \cline{1-10} $2$ & $-10^{-3}$ & $0$ & $0$ & $0$ & $0$ & $0$ & $1$ & $0$ & $-1$ & $b^*_2-b_2 \approx 0.259$ \\
        \cline{1-10} $2$ & $-10^{-3}$ & $0$ & $1$ & $-1$ & $0$ & $0$ & $0$ & $1$ & $-1$ & $b^*_3-b_3 \approx 2.156$ \\
        \hline
        \hline $1$ & $-10^{-3}$ & $-1$ & $1$ & $0$ & $0$ & $0$ & $0$ & $0$ & $0$ & $\{b_4-b_1\} \approx 0.513$ \\
        \cline{1-10} $1$ & $-10^{-3}$ & $0$ & $0$ & $0$ & $0$ & $0$ & $0$ & $0$ & $0$ & $\{b_4-b_2\} \approx 0.995$ \\
        \cline{1-10} $2$ & $-10^{-3}$ & $0$ & $1$ & $-1$ & $0$ & $0$ & $0$ & $0$ & $0$ & $\{b_4-b_3\} \approx 0.194$ \\
        \hline
    \end{tabular}
\end{center}

Notice how all boosts which are still increasing have $10^{-3}$ in the $\epsilon$ column, making all distances in both $D_a$ and $D_c$ have $-10^{-3}$ in this column. This is referred to in invariant (i2), and means that when the smallest value from $D_a \cup D_c$ is added to the three boosts, the $\epsilon$ column becomes zero again, although it may immediately change in case $I_c < I_a$ because then a new epsilon is added. Continuing the example, the smallest value in $D_a \cup D_c$ is given by $\{b_4-b_3\} \approx 0.194$. Adding this together with a new $q \epsilon$ where $q$ can be chosen as $10^{-3}$ again leads to the following boosts, $D_a$, and $D_c$:

\begin{center}
    \begin{tabular}{|>{\boldmath}c>{\boldmath}c|c>{\boldmath}c>{\boldmath}cc|cccc||l|}
        \hline $0$ & $10^{-3}$ & $1$ & $0$ & $-1$ & $0$ & $0$ & $0$ & $0$ & $1$ & $b_1 \approx 12.509$ \\
        \cline{1-10} $0$ & $10^{-3}$ & $0$ & $1$ & $-1$ & $0$ & $0$ & $0$ & $0$ & $1$ & $b_2 \approx 12.027$ \\
        \cline{1-10} $0$ & $10^{-3}$ & $0$ & $0$ & $0$ & $0$ & $0$ & $0$ & $0$ & $1$ & $b_3 \approx 13.828$ \\
        \cline{1-10} \unboldmath{$0$} & \unboldmath{$0$} & \unboldmath{$0$} & \unboldmath{$0$} & \unboldmath{$0$} & \unboldmath{$0$} & \unboldmath{$0$} & \unboldmath{$0$} & \unboldmath{$0$} & \unboldmath{$1$} & $b_4 \approx 13.823$ \\
        \hline
        \hline $0$ & $-10^{-3}$ & $-1$ & $0$ & $1$ & $0$ & $1$ & $0$ & $0$ & $-1$ & $b^*_1-b_1 \approx 0.700$ \\
        \cline{1-10} $0$ & $-10^{-3}$ & $0$ & $-1$ & $1$ & $0$ & $0$ & $1$ & $0$ & $-1$ & $b^*_2-b_2 \approx 0.060$ \\
        \cline{1-10} $0$ & $-10^{-3}$ & $0$ & $0$ & $0$ & $0$ & $0$ & $0$ & $1$ & $-1$ & $b^*_3-b_3 \approx 1.957$ \\
        \hline
        \hline $-1$ & $-10^{-3}$ & $-1$ & $0$ & $1$ & $0$ & $0$ & $0$ & $0$ & $0$ & $\{b_4-b_1\} \approx 0.314$ \\
        \cline{1-10} $-1$ & $-10^{-3}$ & $0$ & $-1$ & $1$ & $0$ & $0$ & $0$ & $0$ & $0$ & $\{b_4-b_2\} \approx 0.796$ \\
        \cline{1-10} $1$ & $-10^{-3}$ & $0$ & $0$ & $0$ & $0$ & $0$ & $0$ & $0$ & $0$ & $\{b_4-b_3\} \approx 0.995$ \\
        \hline
    \end{tabular}
\end{center}

This has increased $\lfloor b_3 - b_4 \rfloor$ from -1 to 0, and decreased $\lfloor b_4 - b_3 \rfloor$ from 0 to -1. This is the case where one boost overtakes another (case $I_c < I_a$ subcase c in the proof) which does not change the policy.

The nearest change is now $b_2$ reaching its target boost $b^*_2$, leading to the following boosts:

\begin{center}
    \begin{tabular}{|>{\boldmath}c>{\boldmath}c|c>{\boldmath}c>{\boldmath}cc|c>{\boldmath}cc>{\boldmath}c||l|}
        \hline $0$ & $0$ & $1$ & $-1$ & $0$ & $0$ & $0$ & $1$ & $0$ & $0$ & $b_1 \approx 12.569$ \\
        \cline{1-10} $0$ & $0$ & $0$ & $0$ & $0$ & $0$ & $0$ & $1$ & $0$ & $0$ & $b_2 \approx 12.087$ \\
        \cline{1-10} $0$ & $0$ & $0$ & $-1$ & $1$ & $0$ & $0$ & $1$ & $0$ & $0$ & $b_3 \approx 13.888$ \\
        \cline{1-10} \unboldmath{$0$} & \unboldmath{$0$} & \unboldmath{$0$} & \unboldmath{$0$} & \unboldmath{$0$} & \unboldmath{$0$} & \unboldmath{$0$} & \unboldmath{$0$} & \unboldmath{$0$} & \unboldmath{$1$} & $b_4 \approx 13.823$ \\
        \hline
    \end{tabular}
\end{center}

Notice how for $b_1, b_2, b_3$, the 1 in the $4^*$ column moves to the $2^*$ column now that $b_2$ reached $b^*(s)$. This is referred to at the end of invariant (i1) as a change that sets the $4^*$ column to zero. From this point, the minimum of $D_a \cup D_c$ is $\{b_2-b_3\} \approx 0.199$, and adding this with an extra $q \epsilon$ with $q = 10^{-3}$ as before gives the following boosts:

\begin{center}
    \begin{tabular}{|>{\boldmath}c>{\boldmath}c|c>{\boldmath}c>{\boldmath}cc|cccc|l|}
        \hline $2$ & $10^{-3}$ & $1$ & $0$ & $-1$ & $0$ & $0$ & $1$ & $0$ & $0$ & $b_1 \approx 12.772$ \\
        \cline{1-10} \unboldmath{$0$} & \unboldmath{$0$} & \unboldmath{$0$} & \unboldmath{$0$} & \unboldmath{$0$} & \unboldmath{$0$} & \unboldmath{$0$} & \unboldmath{$1$} & \unboldmath{$0$} & \unboldmath{$0$} & $b_2 \approx 12.087$ \\
        \cline{1-10} $2$ & $10^{-3}$ & $0$ & $0$ & $0$ & $0$ & $0$ & $1$ & $0$ & $0$ & $b_3 \approx 14.091$ \\
        \cline{1-10} \unboldmath{$0$} & \unboldmath{$0$} & \unboldmath{$0$} & \unboldmath{$0$} & \unboldmath{$0$} & \unboldmath{$0$} & \unboldmath{$0$} & \unboldmath{$0$} & \unboldmath{$0$} & \unboldmath{$1$} & $b_4 \approx 13.823$ \\
        \hline
    \end{tabular}
\end{center}
Both $b_1$ and $b_3$ keep the 1 in the $2^*$ column as the floors that changed involved $b_2$ (which has become equal to $b^*_2$). This change increased $\lfloor b_3 - b_2 \rfloor$ from 1 to 2, and decreased $\lfloor b_2 - b_3 \rfloor$ from -2 to -3. The latter change does not affect the policy, but the former does. As $k = 2 \leq N^*(3, 2) = 3$, Theorem \ref{th:swap} shows that this change improves the prefactor of the policy.

The next change is at a distance $\{b_4-b_1\} \approx 0.051$: this changes a floor involving $b_4$ instead of $b_2$, meaning the boosts $b_1$ and $b_3$ see the 1 in the $2^*$ column move back over to the $4^*$ column. As explained at the end of invariant (i1), this change can happen many times throughout the procedure. The new boosts are given as

\begin{center}
    \begin{tabular}{|>{\boldmath}c>{\boldmath}c|>{\boldmath}cc>{\boldmath}cc|c>{\boldmath}cc>{\boldmath}c|l|}
        \hline $-1$ & $10^{-3}$ & $0$ & $0$ & $0$ & $0$ & $0$ & $0$ & $0$ & $1$ & $b_1 \approx 12.828$ \\
        \cline{1-10} \unboldmath{$0$} & \unboldmath{$0$} & \unboldmath{$0$} & \unboldmath{$0$} & \unboldmath{$0$} & \unboldmath{$0$} & \unboldmath{$0$} & \unboldmath{$1$} & \unboldmath{$0$} & \unboldmath{$0$} & $b_2 \approx 12.087$ \\
        \cline{1-10} $-1$ & $10^{-3}$ & $-1$ & $0$ & $1$ & $0$ & $0$ & $0$ & $0$ & $1$ & $b_3 \approx 14.147$ \\
        \cline{1-10} \unboldmath{$0$} & \unboldmath{$0$} & \unboldmath{$0$} & \unboldmath{$0$} & \unboldmath{$0$} & \unboldmath{$0$} & \unboldmath{$0$} & \unboldmath{$0$} & \unboldmath{$0$} & \unboldmath{$1$} & $b_4 \approx 13.823$ \\
        \hline
    \end{tabular}
\end{center}

This process continues for 6 more steps before all boosts reach their target values:
\begin{enumerate}
    \item $\lfloor b_1 - b_2 \rfloor$ increases from 0 to 1, $\lfloor b_2 - b_1 \rfloor$ decreases from -1 to -2. This is an improvement because $k = 1 \leq N^*(1, 2) = 1$,
    \item $b_1$ reaches $b^*_1$. At this point, only $b_3$ is still increasing.
    \item $\lfloor b_3 - b_4 \rfloor$ increases from 0 to 1, $\lfloor b_4 - b_3 \rfloor$ decreases from -1 to -2. This improves the policy because $k = 1 \leq N^*(3, 4) = 1$,
    \item $\lfloor b_3 - b_2 \rfloor$ increases from 2 to 3, $\lfloor b_2 - b_3 \rfloor$ decreases from -3 to -4. $k = 3 \leq N^*(3, 2) = 3$ so Theorem \ref{th:swap} says it's an improvement.
    \item $\lfloor b_3 - b_1 \rfloor$ increases from 1 to 2, $\lfloor b_1 - b_3 \rfloor$ decreases from -2 to -3. $N^*(3, 1) = 2$ so it is an improvement,
    \item and finally $b_3$ reaches $b^*_3$.
\end{enumerate}

\section{Heavy Traffic} \label{sec:heavytraffic}

In this section we introduce an approximation $b^{heavy}$ for the boost function $b^*$ which may be easier to estimate in
practice. This approximation becomes exact as the load tends to one. Moreover, we show that the prefactor of $\gamma$-CounterBoost converges to that of $\gamma$-Boost as $\lambda$ tends to one, showing that the gains $\gamma$-Boost achieves using arrival time information vanish in heavy traffic. 

\begin{theorem} \label{th:heavy}
For $\lambda$ tending to one we have $b^*_i - b^*_j \approx log(E[S_j]/E[S_i])/log(1+\dr/\lambda)$ and
    \begin{align}
        \lim_{\lambda \rightarrow 1^-} c_{\gamma-CounterBoost} = c_{FCFS} \prod_{j=1}^d E[S_j]^{-p_j E[S_j]},
    \end{align}
    which is the same limit as for $\gamma$-Boost in case of unknown job sizes and $d$ job types.
\end{theorem}
\begin{proof}
    Recall that $\tS{}=1+\dr/\lambda$ for Poisson arrivals \cite{abate1}.
    The key observation is that $\tS{i} \approx 1+\dr E[S_i]$ for $\lambda$ tending to one. This means we have
    \begin{align*}
        \tT{j} = \sum_{k=1}^d p_k \tS{k}^{1(k<j)} \approx 1 + \dr \sum_{k=1}^{j-1} p_k E[S_k].
    \end{align*}
    Using $\lim_{x \rightarrow 0} (1+ax)^{b/\log(1+x)} = e^{ab}$, one finds for $j \leq i$
    \begin{align*}
         \tT{j}^{N^*(j-1,i)-N^*(j,i)} &\approx \left(1 + \dr \sum_{k=1}^{j-1} p_k E[S_k] \right)^{b^*_{j-1} - b^*_j}  \approx \left(1 + \dr \sum_{k=1}^{j-1} p_k E[S_k] \right)^{log(E[S_j]/E[S_{j-1}])/log(1+\dr)} 
         \\&\approx \left(\frac{E[S_j]}{E[S_{j-1}]}\right)^{\sum_{k=1}^{j-1} p_k E[S_k]}
    \end{align*}
    and similarly due to $\tS{}^{b^*_j - b^*_{j-1}} \approx E[S_j]/E[S_{j-1}]$ for $j > i$
    \begin{align*}
        \left(\frac{\tT{j}}{\tS{}}\right)^{N^*(i,j)-N^*(i,j-1)} \approx \left(\frac{E[S_j]}{E[S_{j-1}]}\right)^{-1 + \sum_{k=1}^{j-1} p_k E[S_k]}.
    \end{align*}
    From this we have
    \begin{align*}
    c_{W^{(i)}}(b^*) & \approx c_Z \prod_{j=2}^{i} \left(\frac{E[S_j]}{E[S_{j-1}]}\right)^{\sum_{k=1}^{j-1} p_k E[S_k]} \prod_{j=i+1}^{d} \left(\frac{E[S_j]}{E[S_{j-1}]}\right)^{-1 + \sum_{k=1}^{j-1} p_k E[S_k]} \\
    &= c_Z \prod_{j=2}^{d} \left(\frac{E[S_j]}{E[S_{j-1}]}\right)^{\sum_{k=1}^{j-1} p_k E[S_k]} \prod_{j=i+1}^{d} \left(\frac{E[S_{j-1}]}{E[S_j]}\right) = c_Z E[S_i] \prod_{j=1}^d E[S_j]^{-p_j E[S_j]}.
    \end{align*}
    which we can then use to simplify $c_{\gamma-CounterBoost}$:
    \begin{align*}
        c_{\gamma-CounterBoost} &= \sum_{i=1}^d p_i \tS{i} c_{W^{(i)}}(b^*) \approx \sum_{i=1}^d p_i \tS{i} c_Z E[S_i] \prod_{j=1}^d E[S_j]^{-p_j E[S_j]} \\
        &\approx c_Z \sum_{i=1}^d p_i  E[S_i] \prod_{j=1}^d E[S_j]^{-p_j E[S_j]} \approx c_{FCFS} \prod_{j=1}^d E[S_j]^{-p_j E[S_j]},
    \end{align*}
    where $c_Z \approx c_{FCFS} = c_Z \tS{}$ as $\lim_{\gamma \rightarrow 0} \tS{} = 1$, and similarly $\lim_{\gamma \rightarrow 0} \tS{i} = 1$, which is why this factor can be removed going from the first to the second line.
\end{proof}

In practice, it may be difficult to estimate the optimal boosts $b^*$ accurately as this relies on the decay rate and the Laplace transforms of the job size distributions. A simplified expression for the boost in the heavy traffic limit 
is easier to estimate and has a similar tail as $\gamma$-CounterBoost for high loads:
\begin{align}
    b^*_i &= \log \left( \frac{\tS{i}}{\tS{i}-1} \right) / \log \tS{} \approx \log \left( \frac{1 + \dr E[S_i]}{\dr E[S_i]} \right) / \log (1 + \dr) \nonumber \\
    &\approx \log \left( 1 + \frac{1}{\dr E[S_i]} \right) / \dr \approx \frac{-\log \dr - \log E[S_i]}{\dr}.  \label{eq:bheavy_nokingman}
\end{align}
The expression for $b^{heavy}_i$ is then found by adding a term $\log (\dr) / \dr$ to each boost (which does not change the policy), and using the Kingman heavy traffic approximation \cite{kingman62} for the decay rate $\dr \approx \frac{2(1-\lambda)}{E[S^2]}$:
\begin{align}
    b^{heavy}_i = \frac{E[S^2]}{2(\lambda-1)} \log E[S_i]. \label{eq:bheavy}
\end{align}

In the next section we show how this approximation for $b^*_i$ affects the tail constant for two particular settings.

\section{Numerical Results}\label{sec:num}

In this section we show some numerical results, comparing $\gamma$-CounterBoost to other policies such as $\gamma$-Boost and Nudge-$M$, or to CounterBoost policies with approximated optimal boosts, such as the heavy traffic approximation from Section \ref{sec:heavytraffic}.
We show what happens with two different somewhat arbitrarily chosen settings (a and b), each with three job types.
Whenever we use a hyperexponential (HE) distribution, we specify its mean, the squared coefficient of variation ($SCV$), and the shape parameter $f$, which indicates what fraction of the work is contributed by the first phase of the hyperexponential distribution (with a fraction $1-f$ contributed by the second phase). The settings are as follows:
\begin{enumerate}
    \item[a.] The three types occur with probabilities $p_1= 1/8$, $p_2=3/4$, and $p_3=1/8$ respectively. Their sizes are hyperexponentially distributed with two phases. The means are $1/5$, $1$, and $9/5$, the SCVs are $5$, $10$, and $3$, and the shapes $f$ are $9/10$, $5/10$, and $7/10$.
    \item[b.] The probabilities are $p_1 = 1/5$, $p_2 = 3/4$, $p_3 = 1/20$. Type-1 jobs are exponentially distributed with mean $1/10$, type-2 jobs follow an Erlang distribution with two phases and mean $9/10$, and type-3 jobs follow a hyperexponential distribution with mean $61/10$, $SCV = 5$, and shape $f = 7/10$.
\end{enumerate}
All plots show the ratio between the tail constant of these policies over FCFS.

\begin{figure*}[t!]
\begin{subfigure}[t]{.48\textwidth}
  \centering
  \includegraphics[width=\linewidth]{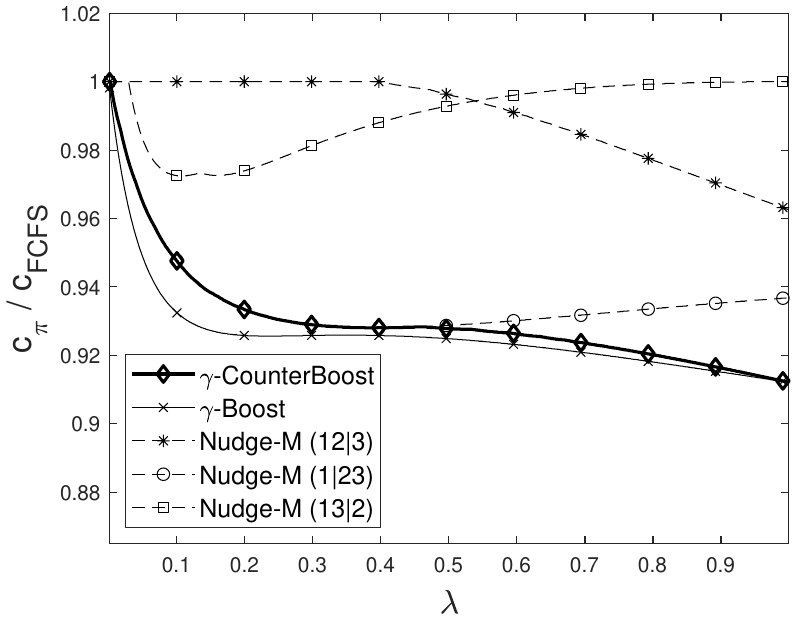}
  \caption{Hyperexponential job sizes}\label{fig:1a}
\end{subfigure}
\begin{subfigure}[t]{.48\textwidth}
  \centering
  \includegraphics[width=\linewidth]{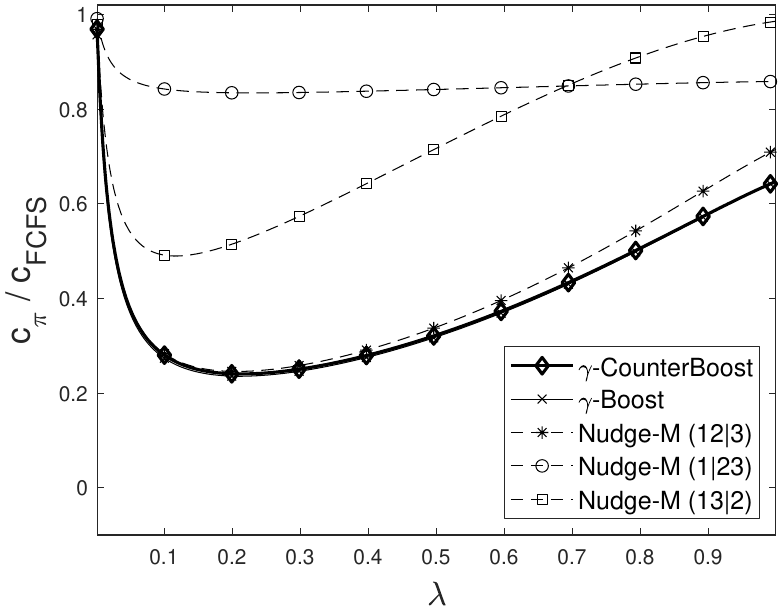}
  \caption{Mix of job size distributions}\label{fig:1b}
\end{subfigure}

\caption{The ratio between the tail constant of $\gamma$-Boost, $\gamma$-CounterBoost, and Nudge-$M$ over FCFS. For Nudge-$M$, three curves show the three different ways to group two of the three types together, as that policy requires two job types.}
\label{fig:1}
\end{figure*}

Figure \ref{fig:1} shows that $\gamma$-Boost is indeed somewhat better than the other policies especially for lower loads, but as the load tends to 1, $\gamma$-CounterBoost reaches the same tail constant, as proven in Section \ref{sec:heavytraffic}. It is also visible that the tail constant of Nudge-$M$ depends strongly on which two types are joined together. In Figure \ref{fig:1a}, for lower loads merging types 2 and 3 gives the same constant as $\gamma$-CounterBoost while for higher loads, keeping all three types distinct is crucial. The other two ways to merge two types are less effective in this setting. In Figure \ref{fig:1b}, $\gamma$-CounterBoost has nearly the same tail constant as $\gamma$-Boost, the difference is too small to see on the plot. In this setting, merging types 1 and 2 together does not increase the tail constant that much for higher loads. Further note that the gain over FCFS is much more
substantial in setting b, where $\gamma$-Boost and $\gamma$-CounterBoost nearly coincide for all loads.

Figure \ref{fig:2} shows the same ratio $c_\pi / c_{FCFS}$ as before, now comparing $\gamma$-CounterBoost to a CounterBoost policy where the boosts are chosen 30\% or 50\% above and below the optimal boosts. We see that the tail constant is not very sensitive to these rather large deviations from the optimal boosts in the two settings considered. We note that this figure, as well as Figure \ref{fig:3} contain non-smooth lines. These are explained by the policy not changing continuously, but rather changing in steps each time an $N_b(i, j)$ changes. For higher loads the boosts are larger and 
therefore the non-smoothness is less pronounced as 
for instance increasing $N_b(i, j)$ from $100$ to $101$ has far less impact that changing it from $3$ to $4$.

\begin{figure*}[t!]
\begin{subfigure}[t]{.48\textwidth}
  \centering
  \includegraphics[width=\linewidth]{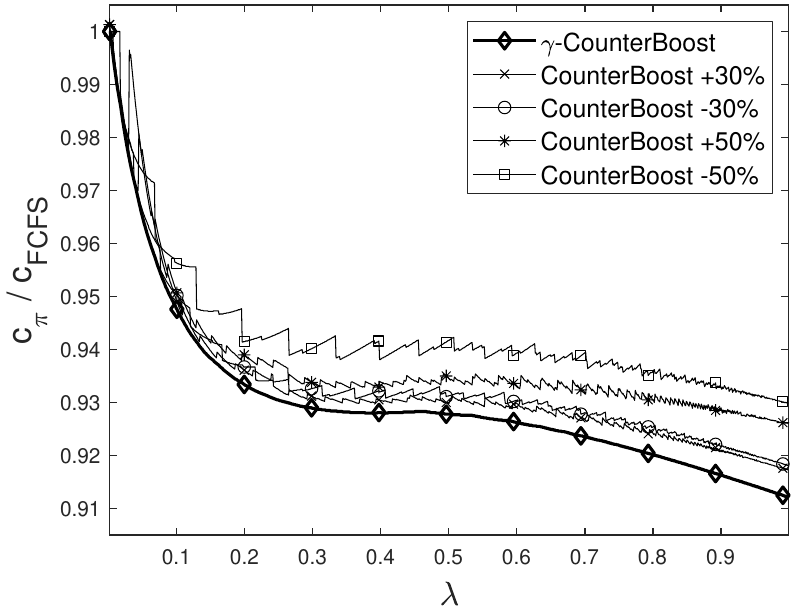}
  \caption{Hyperexponential job sizes} \label{fig:2a}
\end{subfigure}
\begin{subfigure}[t]{.48\textwidth}
  \centering
  \includegraphics[width=\linewidth]{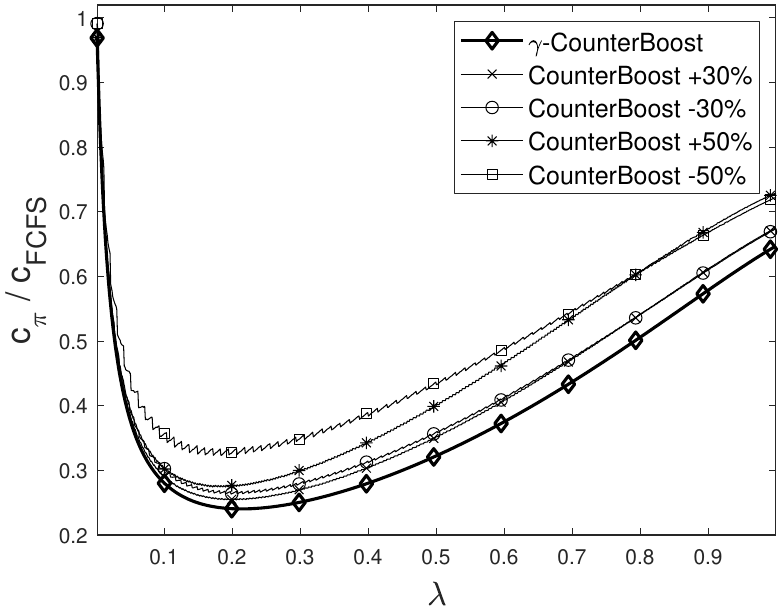}
  \caption{Mix of job size distributions} \label{fig:2b}
\end{subfigure}

\caption{The ratio between the tail constant of $\gamma$-CounterBoost over FCFS, compared to this ratio when over- or underestimating the boosts by $30\%$.}
\label{fig:2}
\end{figure*}

Finally we look at  the tail constant of CounterBoost when $b^*$ is replaced by its heavy traffic approximation 
$b^{heavy}$ in Figure \ref{fig:3}. The figure illustrates that using $b^{heavy}$ yields most of the improvement when the load is reasonably high. The figure also shows the tail constant when using a variation on $b^{heavy}$, where we use the exact decay rate 
$\gamma$ instead of the Kingman approximation $2(1-\lambda)/E[S^2]$, i.e., using \eqref{eq:bheavy_nokingman} instead of \eqref{eq:bheavy}. In Figure \ref{fig:3a}, using the exact decay rate is actually worse than using the approximation, while in Figure \ref{fig:3b}, using the correct decay rate instead of the approximation is better for most loads.
This figure also shows CounterBoost with boosts $b^{expo}$: these are the optimal boosts for a queue where the probabilities 
$p_1, p_2$ and $p_3$ as well as the per type mean job sizes are the same as before, but the boosts are computed
assuming type-$1$, $2$ and $3$ job sizes are exponential. The boost $b^{expo}$ only requires estimating the 
probabilities $p_i$ and mean job size $E[S_i]$
per type, but as shown in Figure \ref{fig:3a}, the computed boost values are not particularly good. 
These boosts achieve typically less than half of the improvement over FCFS compared to $\gamma$-CounterBoost in both settings.

\begin{figure*}[t!]
\begin{subfigure}[t]{.48\textwidth}
  \centering
  \includegraphics[width=1\linewidth]{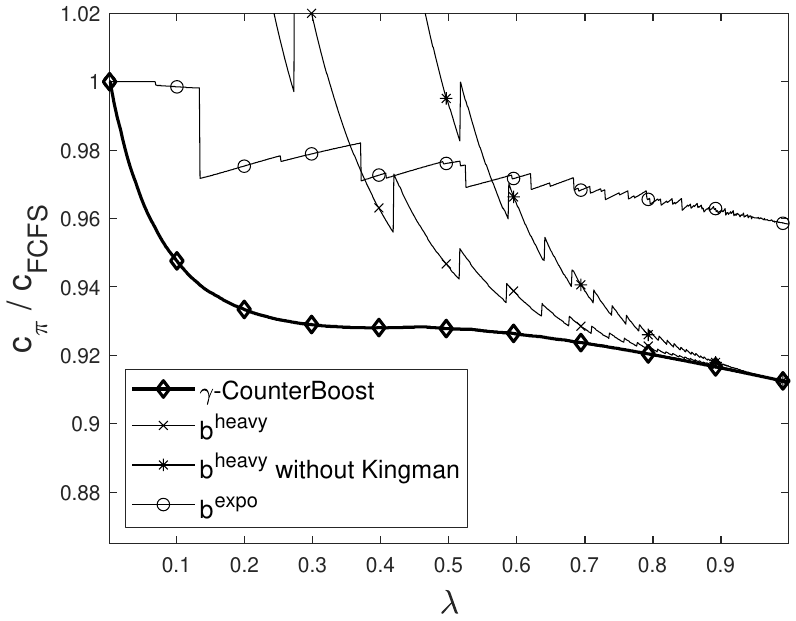}
	\caption{Hyperexponential job sizes}
	\label{fig:3a}
\end{subfigure}
\begin{subfigure}[t]{.48\textwidth}
  \centering
  \includegraphics[width=1\linewidth]{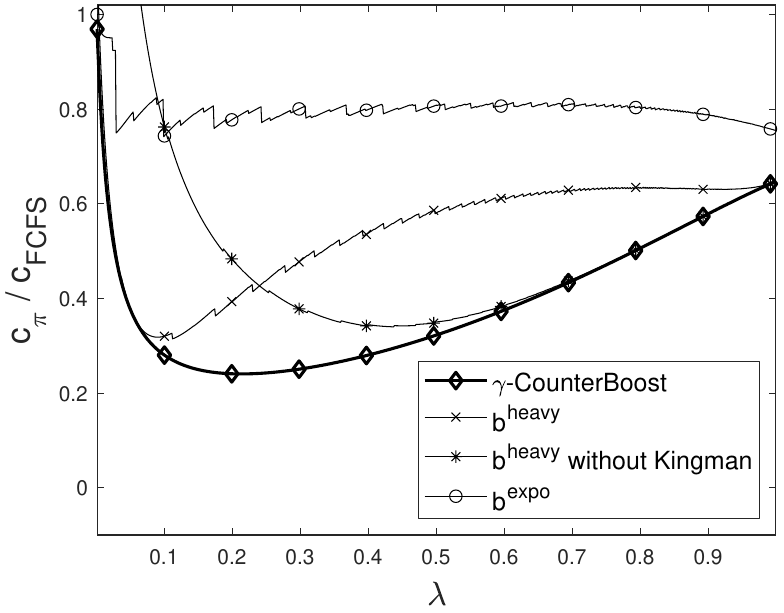}
	\caption{Mix of job size distributions}
	\label{fig:3b}
\end{subfigure}
\caption{The ratio between the prefactor of CounterBoost using $b^*$ and $b^{heavy}$ respectively over FCFS.}
\label{fig:3}
\end{figure*}

\section{Conclusion} \label{sec:conc}
The $\gamma$-CounterBoost policy was shown to be tail optimal within the class of Contextual CounterBoost policies in a partial information setting. We conjecture that in this setting, $\gamma$-CounterBoost is tail optimal among not just Contextual CounterBoost policies, but among all policies. In the partial information setting considered the scheduler knows the type of incoming jobs, but not their exact sizes. $\gamma$-CounterBoost schedules jobs based on a boosted arrival counter, where the boost value only depends on the type of the incoming job. While in general $\gamma$-CounterBoost has an inferior tail\footnote{This was proven for $d=2$ types in \cite{boostZiv}, but the result should generalize to any finite number of job types $d$.} compared to the $\gamma$-Boost policy introduced in \cite{boostZiv}, which boosts arrival {\it times} as opposed to {\it counters},  both tails were shown to coincide in the heavy traffic limit. %
Numerical results further showed that the gain of $\gamma$-Boost over $\gamma$-CounterBoost can be close to zero for all arrival rates in some particular settings.

As future work one could also try to prove optimality
of $\gamma$-Boost in the partial information setting within a class of contextual policies, that is, 
if the arrival time boost of a job not only depends on the job type, but also the types of the previous
and next $M$ job arrivals. One could also try to generalize the result in this paper to a full information setting by considering infinitely many job types, where the counter boost of a job depends on its size.

\section*{Acknowledgement}
We would like to thank George Yu and Ziv Scully for some useful discussions and for introducing the class of 
Contextual CounterBoost policies. This work was supported by the FWO project G0A9823N.

\bibliographystyle{IEEEtran}
\bibliography{thesis}

\end{document}